\documentclass[a4paper,USenglish,cleveref,autoref,numberwithinsect]{lipics-v2021}
\pdfoutput=1 

\usepackage[T5,T1]{fontenc} 

\hideLIPIcs  

\title{Syntactically and semantically regular languages of \pdflambda-terms coincide through logical relations}

\titlerunning{Syntactically and semantically regular languages of \pdflambda-terms coincide} 

\author{Vincent Moreau}{IRIF \& Université Paris Cité \& Inria Paris, France \and \url{http://www.irif.fr/~moreau}}{moreau@irif.fr}{https://orcid.org/0009-0005-0638-1363}{}

\author{Lê Thành Dũng (Tito) Nguy\~{\^e}n}{Laboratoire de l'informatique du parallélisme (LIP), École normale supérieure de Lyon, France \and \url{https://nguyentito.eu/}}{nltd@nguyentito.eu}{https://orcid.org/0000-0002-6900-5577}{Supported by the LABEX MILYON (ANR-10-LABX-0070) of Université de Lyon, within the program \enquote{Investissements d'Avenir} (ANR-11-IDEX-0007) operated by the French National Research Agency (ANR).}

\authorrunning{V.~Moreau and L.~T.~D.~Nguy\~{\^e}n} 

\Copyright{Vincent Moreau and Lê Thành Dũng Nguy\~{\^e}n} 



\category{} 

\relatedversiondetails[linktext={doi:10.4230/LIPIcs.CSL.2024.40}]{Official conference proceedings version}{https://dx.doi.org/10.4230/LIPIcs.CSL.2024.40}

\acknowledgements{We would like to thank Amina Doumane, Sam van Gool, Paul-André Melliès and Sylvain Salvati for in-depth discussions that significantly helped us refine our ideas. We are also grateful to Sam and Paul-André for proof-reading drafts of this paper, and to the ReFL discussion group \url{https://www.engboris.fr/refl/} for hosting a reading group on logical relations and normalization by evaluation. The first author would like to thank the Felicissimo family for their support during the writing process of this article.}

\nolinenumbers 

\EventEditors{Aniello Murano and Alexandra Silva}
\EventNoEds{2}
\EventLongTitle{32nd EACSL Annual Conference on Computer Science Logic (CSL 2024)}
\EventShortTitle{CSL 2024}
\EventAcronym{CSL}
\EventYear{2024}
\EventDate{February 19--23, 2024}
\EventLocation{Naples, Italy}
\EventLogo{}
\SeriesVolume{288}
\ArticleNo{12}

\usepackage{csquotes}
\usepackage{xspace}
\makeatletter
\xspaceaddexceptions{\csq@qclose@i}
\makeatletter

\usepackage{macros}

\usepackage{tikz}
\usetikzlibrary{arrows,automata,positioning}
\usepackage{adjustbox} 
 
\usepackage{amsmath}
\usepackage{amssymb}
\usepackage{mathtools}

\usepackage{todonotes}

\usepackage{apxproof}

\begin{document}

\maketitle

\begin{abstract}
    A fundamental theme in automata theory is regular languages of words and trees, and their many equivalent definitions. Salvati has proposed a generalization to regular languages of simply typed~$\lambda$-terms, defined using denotational semantics in finite sets.

    We provide here some evidence for its robustness. First, we give an equivalent syntactic characterization that naturally extends the seminal work of Hillebrand and Kanellakis connecting regular languages of words and syntactic $\lambda$-definability. Second, we show that any finitary extensional model of the simply typed $\lambda$-calculus, when used in Salvati's definition, recognizes exactly the same class of languages of $\lambda$-terms as the category of finite sets does.

    The proofs of these two results rely on logical relations and can be seen as instances of a more general construction of a categorical nature, inspired by previous categorical accounts of logical relations using the gluing construction.
\end{abstract}

\section{Introduction}
\label{sec:introduction}

Much work has been devoted to the study of regular languages of words and trees -- also called recognizable -- and their equivalent characterizations, typically in terms of automata, algebra, and logic. The remarkable robustness of this notion of regularity has led to attempts to extend it to several other structures, such as infinite words/trees or graphs of bounded treewidth -- many examples can be found, for instance, in~\cite{DBLP:journals/corr/abs-2008-11635}.

This paper focuses on a less studied extension: recognizable languages of simply typed $\lambda$-terms, introduced by Salvati~\cite{DBLP:conf/wollic/Salvati09}. They are a conservative generalization~\cite[\S3]{DBLP:conf/wollic/Salvati09}:\footnote{Alternatively, see \cite[Proposition~7.1]{entics:12280} for the case of words.} using the "Church encoding", finite words and trees can be represented in the simply typed $\lambda$-calculus, and recognizability for Church encodings coincides with the usual notion of regularity. Furthermore, as Salvati explains in his habilitation thesis~\cite{DBLP:books/hal/Salvati15}, regular languages of $\lambda$-terms can be used to shed light on several classical topics concerning the simply typed $\lambda$-calculus, such as higher-order matching~\cite[\S6.2]{DBLP:conf/wollic/Salvati09} and higher-order grammars~\cite{DBLP:journals/iandc/KobeleS15}.

However, not many characterizations of the class of recognizable languages of simply typed $\lambda$-terms are known, and it would be desirable to have more evidence that it is robust. Moreover, if we know that this class of languages is a truly canonical object, then so is its Stone dual, namely the recently introduced space of profinite $\lambda$-terms~\cite{entics:12280}.
Currently, there exist two definitions of recognizable languages of $\lambda$-terms, both provided in~\cite{DBLP:conf/wollic/Salvati09}:
\begin{itemize}
    \item The first one~\cite[Definition~1]{DBLP:conf/wollic/Salvati09} uses denotational semantics into the finite standard models of the simply typed $\lambda$-calculus. This may be understood as generalizing to higher orders the computational aspects of deterministic finite automata. In the same vein, the concrete construction of profinite $\lambda$-terms in~\cite{entics:12280} depends on specific properties of the category~$\FinSet$ of finite sets and functions between them.
    \item The second one, grounded in intersection types~\cite[\S4]{DBLP:conf/wollic/Salvati09}, turns out to admit an equivalent presentation in terms of a denotational semantics in finite domains~\cite[Theorem~25]{DBLP:books/hal/Salvati15}. Indeed, the connection between intersection types and semantics is standard, see e.g.~\cite{ronchidellarocca:LIPIcs:2018:9861}.
\end{itemize}

Both definitions can be seen as an instance of the following pattern: the interpretation of a simply typed $\lambda$-term in some denotational model with a \textit{finitary} flavor should suffice to know whether the term belongs to our language of interest. This is closely analogous to the algebraic definition of regular word languages. Indeed, viewing the set $\Sigma^*$ of words over a finite alphabet $\Sigma$ as the free monoid generated by $\Sigma$, a language $L \subseteq \Sigma^*$ is regular if and only if, for some homomorphism $\varphi\colon\Sigma^*\to M$ to a finite monoid $M$, the \enquote{interpretation} $\varphi(w)$ determines\footnote{More formally, $L = \varphi^{-1}(P)$ for some $P \subseteq M$.} whether $w\in L$ for each $w\in\Sigma^*$. Asking for $\varphi$ to be a homomorphism parallels the compositionality property of denotational semantics.

We may therefore ask:
\begin{itemize}
    \item What kind of semantics yield the same notion of recognizable language? More precisely, with a definition of \salvatiReg{\C} language for any cartesian closed category (CCC) $\C$, i.e.\ any categorical model of the simply typed $\lambda$-calculus with products, the question becomes: when do recognizability by $\C$ and by $\FinSet$ coincide?
    \item Alternatively, is there any characterization of regular languages of $\lambda$-terms that does not involve denotational semantics?
\end{itemize}
For the latter, we propose a positive answer inspired by the following result of Hillebrand and Kanellakis~\cite[Theorem~3.4]{DBLP:conf/lics/HillebrandK96}: a language of words is regular if and only if it can be decided by a simply typed $\lambda$-term operating on Church-encoded words. By replacing the type of Church encodings with any simple type $A$, we get a natural notion of \hkReg language of $\lambda$-terms, defined by means purely internal to the simply typed $\lambda$-calculus.

\subparagraph*{Contributions and proof methods.}

The main results of this paper are as follows:
\begin{description}
    \item[\Cref{thm:semantic-evaluation}] any "non-degenerate" cartesian closed category can recognize at least every language which is \hkReg.
    \item[\Cref{thm:salvati-points-finset}] every language recognized by a "locally finite" and "well-pointed" CCC -- in other words, a finitary extensional model of the simply typed $\lambda$-calculus -- is recognized by $\FinSet$. This was first stated by Salvati without proof in~\cite[Lemma~20]{DBLP:books/hal/Salvati15}.
    \item[\Cref{thm:hk-squeezing}] every language recognized by $\FinSet$ is \hkReg.
\end{description}
These three theorems, taken together, show that all "non-degenerate", "well-pointed" and "locally finite" CCCs yield the same notion of regular language of $\lambda$-terms, which is the same as the syntactic one.

To achieve this goal, we introduce a construction on CCCs, which we call "squeezing", and combine it with the standard categorical account of logical relations based on sconing. Indeed, Salvati claims in~\cite{DBLP:books/hal/Salvati15} that \Cref{thm:salvati-points-finset} can be established via logical relations, and it turns out that this falls out directly from our squeezing construction; but its versatility also allows us to apply it to prove the more diffcult \Cref{thm:hk-squeezing} on syntactic recognizability.

\subparagraph*{Related work.}

Morally, the study of regular languages of $\lambda$-terms amounts to understanding what information can be extracted by evaluating simply typed $\lambda$-terms in finitary models. A seminal result in this spirit is Statman's finite completeness theorem~\cite{DBLP:journals/jsyml/Statman82}, which can be rephrased as the regularity of all singleton languages of $\lambda$-terms -- a perspective that has led to a simplified proof~\cite{DBLP:journals/ipl/SrivathsanW12}. The idea of using another CCC than $\FinSet$, easier to use to show Statman's theorem, has been exploited in \cite{DBLP:journals/iandc/KobeleS15}. This shows the advantage to use an appropriate CCC to recognize a given language, a possibility which is extended to a vast class of CCCs in this paper (see \Cref{prop:regular-languages-product} for an example of application).

Finitary semantics are powerful tools, in particular, for understanding the computational power of the simply typed $\lambda$-calculus. For instance, in~\cite{DBLP:conf/lics/HillebrandK96}, Hillebrand and Kanellakis use the finite set semantics to prove their aforementioned theorem on regular word languages; as for finite Scott domains presented as intersection types, they have been applied by Terui~\cite{Terui} to study the complexity of normalizing simply typed $\lambda$-terms.

As can be seen \textit{inter alia} from rather surprising results of Statman~\cite{DBLP:journals/apal/Statman04} and Plotkin~\cite{DBLP:journals/corr/abs-2206-08413}, finitary models are also useful to tame the infinitary aspects of an extension of the simply typed $\lambda$-calculus with a fixed-point operator, called the $\lambda Y$-calculus.
Furthermore, the well-studied higher-order model checking problem is about testing regular properties on infinite trees that can be Church-encoded in the $\lambda Y$-calculus; it sits at the interface between automata and programming languages, with applications to the formal verification of functional programs (see e.g.~\cite{DBLP:conf/ppdp/Kobayashi19}). Decidability of higher-order model checking, first established by Ong through game semantics~\cite{DBLP:conf/lics/Ong06}, now admits proofs based on intersection types~\cite{DBLP:journals/jacm/Kobayashi13,DBLP:conf/lics/Ong15} and on finitary semantics~\cite{DBLP:journals/siglog/Walukiewicz16,grellois:tel-01311150}. Drawing on this line of work, higher-order parity automata~\cite{HOParity} generalize to $\lambda Y$-terms the recognizable languages of simply typed $\lambda$-terms.

The theme of syntactic recognizability \textit{à la} Hillebrand and Kanellakis, for its part, has been recently revived in Nguy\~{\^e}n and Pradic's implicit\footnote{The name is a nod to implicit computational complexity, a field concerned with alternative definitions of complexity classes that avoid low-level machine models and explicit resource bounds. As an example, in addition to their result on regular word languages in the simply typed $\lambda$-calculus, Hillebrand and Kanellakis's paper~\cite{DBLP:conf/lics/HillebrandK96} also contains characterizations of the $k$-EXPTIME and $k$-EXPSPACE hierarchies based on $\lambda$-terms, that are again proved by evaluation in finite sets.} automata theory. They use substructural $\lambda$-calculi and "Church encodings" to characterize star-free languages~\cite{iatlc1} and classes of string-to-string functions computed by transducers~\cite{titoPhD}.

\subparagraph*{Plan of the paper.}

We start by recalling in \Cref{sec:recognition} the semantics of the simply typed $\lambda$-calculus, the notion of language recognized by a CCC as defined in \cite{DBLP:conf/wollic/Salvati09} for finite sets, and by introducing the notion of \hkReg language, generalizing recognition as defined in \cite{DBLP:conf/lics/HillebrandK96}. In \Cref{thm:semantic-evaluation} of \Cref{sec:semantic-evaluation}, we show that every "non-degenerate" CCC recognizes all \hkReg languages. In \Cref{sec:relations-squeezing}, we recall the definition of logical relations and introduce the squeezing construction~$\Sqz{-}$ which will be a crucial tool for the two next sections. In \Cref{sec:partial-surjections}, we recall  the definition of "locally finite" and "well-pointed" CCCs, and show in \Cref{thm:salvati-points-finset} that CCCs enjoying both conditions do not recognize more languages than finite sets do. In \Cref{thm:hk-squeezing} of \Cref{sec:fin-encoding}, we show that languages recognized by finite sets are \hkReg. We finish this paper by giving some consequences of the equivalence established by these three theorems in \Cref{sec:regular-languages}.

\section{Languages of \pdflambda-terms}
\label{sec:recognition}

\paragraph*{Syntax and semantics}

\AP
We first specify the syntax we are working with. The grammars of types and preterms are
\[
    A, B
    ::=
    \tyo
    \mid
    A \tto B
    \mid
    A \times B
    \mid
    1
    \quand
    t, u
    ::=
    x
    \mid
    \lambda(x : A).\,t
    \mid
    t\ u
    \mid
    \langle t, u\rangle
    \mid
    t_i \text{ for $i = 1, 2$}
\]
and we consider the usual typing rules and $\beta\eta$-conversion rules, see e.g.~\cite[\S4.1]{DBLP:books/daglib/0093287-amadio-curien}. We extend the notation of~$t_i$, for the projection to the $i^\text{th}$ coordinate, to the case where $t$ is of type $A_1 \times \dots \times A_n$ and $i$ is between $1$ and $n$. As the $\lambda$-abstractions are annotated, a closed $\lambda$-term has at most one type derivation. For any simple type $A$, we write~$\intro*\Tm{A}$ for the set of closed simply typed $\lambda$-terms of type $A$, taken modulo~$\beta\eta$-conversion. 

We recall the semantics of the simply typed $\lambda$-calculus into cartesian closed categories, abbreviated as CCC, see \cite[Chapter~4]{DBLP:books/daglib/0093287-amadio-curien} for more details. For any CCC $\C$, object~$c$ of~$\C$ and simple type~$A$, we define an object~$\semty{A}{c}$ of~$\C$ by induction on~$A$ as follows:
\[
    \semty{\tyo}{c}
    \ :=\ 
    c
    \qquad
    \semty{A \To B}{c}
    \ :=\ 
    \semty{A}{c} \To \semty{B}{c}
    \qquad
    \semty{A \times B}{c}
    \ :=\ 
    \semty{A}{c} \times \semty{B}{c}
    \qquad
    \semty{1}{c}
    \ :=\ 
    1
\]
Using the CCC structure of $\C$, one can define a family of set-theoretic functions
\[
    \semtm{{-}}{c}
    \quad:\quad
    \Tm{A}
    \ \longto\ 
    \C(1, \semty{A}{c})
    \qquad
    \text{for every simple type~$A$}
\]
called semantic brackets, sending closed~$\lambda$-terms to points of the objects~$\semty{A}{c}$.

These assignments can be described in another way. Let~$\Lam$ be the category whose objects are simple types and whose set of morphisms from~$A$ to~$B$ is~$\Tm{A \To B}$, with the expected composition. This category is the free CCC on one object, i.e., for every CCC~$\C$ and object~$c$ of~$\C$, there exists a unique CCC functor~$\semtm{-}{c} : \Lam \to \C$ such that~$\semtm{\tyo}{c}=c$. This can be represented by the commutativity of the following diagram:
\begin{equation}
    \label{diag:lam-universal-property}
\begin{tikzcd}[ampersand replacement=\&]
    \Lam \\
    1 \& \C
    \arrow["c"', from=2-1, to=2-2]
    \arrow["\tyo", from=2-1, to=1-1]
    \arrow["{\semtm{-}{c}}", dashed, from=1-1, to=2-2]
\end{tikzcd}
\end{equation}

In this paper, the CCCs come with specified terminal object, cartesian products, and exponentials, and CCC functors are required to respect these structures strictly, \emph{on the nose}. In that way, the unicity in the universal property of $\Lam$ depicted in \Cref{diag:lam-universal-property} holds up to equality, and not merely isomorphism.

We write $\FinSet$ for the cartesian closed category of finite sets. The semantics of the simply typed $\lambda$-calculus in this CCC corresponds to its naive set-theoretic interpretation. For ease of notation, we identify the finite set $Q$ with the set of functions $\FinSet(1, Q)$.

\paragraph*{Recognizable languages of~\pdflambda-terms, semantically}

We now define the notion of \salvatiReg{\C} language of $\lambda$-terms, for any CCC $\C$. The case $\C = \FinSet$ corresponds to the notion of regular language of simply typed $\lambda$-terms introduced in \cite[Definition~1]{DBLP:conf/wollic/Salvati09}.

\begin{definition}
    \label{def:reg-sem}\AP
    Let~$\C$ be a CCC and~$c$ be an object of~$\C$. For every simple type~$A$ and subset~$F \subseteq \C(1, \semty{A}{c})$, the language~$\Lang{F}$ of~$\lambda$-terms of type~$A$ is defined as
    \[
        \intro*\Lang{F}
        \quad:=\quad
        \{t \in \Tm{A} \mid \semtm{t}{c} \in F\}
        \ .
    \]
    We define the set~$\Rec[c]{A}$ of languages of~$\lambda$-terms of type~$A$ recognized by~$c$ as
    \[
        \intro*\Rec[c]{A}
        \quad:=\quad
        \{\Lang{F} : F \subseteq \C(1, \semty{A}{c})\}
        \ .
    \]
    Finally, a language $L$ of~$\lambda$-terms of type~$A$ is \intro*\salvatiReg{\C} if there exists an object $c$ of $\C$ such that $L$ belongs to the set $\Rec[c]{A}$.
\end{definition}

\begin{example}
    \label{ex:church-encoding}\AP
    For any natural number $n$, we define the associated simple type
    \[
        \intro*\Church{n}
        \ :=\ (\tyo \tto \tyo)^n \tto \tyo \tto \tyo
        \ .
    \]
    There is a bijection between the sets~$\Tm{\Church{n}}$ and $\{1, \dots, n\}^*$, the set of finite words over an alphabet with $n$ letters, called the ""Church encoding"". For example, the word $12212$ over the two letter-alphabet $\{1, 2\}$ is encoded as the $\lambda$-term
    \[
        \lambda (a : (\tyo \tto \tyo)^2).\,\lambda(e : \tyo).\ a_2\ (a_1\ (a_2\ (a_2\ (a_1\ e))))
        \quad\in\quad\Tm{\Church{2}}
        \ .
    \]
    Under this bijection, a language of $\lambda$-terms of type $\Church{n}$ is \salvatiReg{\FinSet}, in the sense of \Cref{def:reg-sem}, if and only if the language of words associated by the "Church encoding" is a regular language of finite words, see \cite[Proposition~7.1]{entics:12280}.
\end{example}

\begin{example}
    \label{ex:odd-salvati}\AP
    We give a detailed example using the "Church encoding". We show that the language $L$ of $\lambda$-terms of type $\Church{2}$ which are encodings of words in $\{1,2\}^*$ that contain an even number of 1s and an odd number of 2s, is \salvatiReg{\FinSet}.
    
    Let $Q$ be the finite set $\{\stateTrue, \stateFalse\}$ and $F$ be the subset of $\semty{\Church{2}}{Q}$ defined as
    \[
        F
        \ :=\ 
        \{
            f \in (Q \tto Q) \times(Q \tto Q) \tto Q \tto Q
            \mid
            f(\swap,\Id_Q)(\stateTrue) = f(\Id_Q,\swap)(\stateFalse) = \stateTrue
        \}
    \]
    where $\swap : Q \to Q$ is the function defined as $\swap(\stateTrue) = \stateFalse$ and $\swap(\stateFalse) = \stateTrue$. The language $L$ is equal to $\Lang{F}$ which belongs to $\Rec[Q]{\Church{2}}$, so $L$ is \salvatiReg{\FinSet}.

    The idea is that, given the semantic interpretation $f \in \semty{\Church2}{Q}$ of the encoding of a word $w \in \{1,2\}^*$, the states $f(\swap,\Id_Q)(\stateTrue)$ and $f(\Id_Q, \swap)(\stateFalse)$ are the states reached respectively, after reading $w$, in the two following deterministic finite automata:
    \[
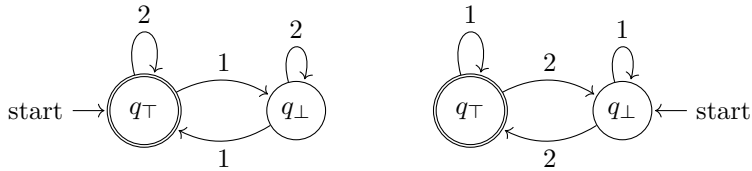
        \begin{adjustbox}{valign=M}
            \begin{tikzpicture}[shorten >=1pt,node distance=2cm,on grid,auto]
            \tikzset{every state/.style={minimum size=3pt}} 
          
            \node[state, initial, accepting]   (qbot) {$\;\stateTrue\;$};
            \node[state] (qtop) [right of=qbot] {$\stateFalse$};
          
            \path[->]
            (qbot) edge [bend left] node {1}  (qtop)
            (qtop) edge [bend left] node {1}  (qbot)
            (qbot) edge [loop above] node {2}  (qbot)
            (qtop) edge [loop above] node {2}  (qtop);
        \end{tikzpicture}
    \end{adjustbox}
        \qquad\qquad
        \begin{adjustbox}{valign=M}
            \begin{tikzpicture}[shorten >=1pt,node distance=2cm,on grid,auto]
            \tikzset{every state/.style={minimum size=3pt}} 
          
            \node[state, accepting]   (qbot) {$\;\stateTrue\;$};
            \node[state, initial right] (qtop) [right of=qbot] {$\stateFalse$};
          
            \path[->]
            (qbot) edge [bend left] node {2}  (qtop)
            (qtop) edge [bend left] node {2}  (qbot)
            (qbot) edge [loop above] node {1}  (qbot)
            (qtop) edge [loop above] node {1}  (qtop);
        \end{tikzpicture}
    \end{adjustbox}
    \]
    The language $L$ is the intersection of the two languages recognized by these automata.
\end{example}

\begin{example}
    \label{ex:affine-terms}\AP
We consider the simple type
\[
    \intro*\UntypedTerms
    \quad:=\quad
    ((\tyo \tto \tyo) \tto \tyo) \tto (\tyo \tto \tyo \tto \tyo) \tto \tyo
    \ .
\]
There is a canonical bijection -- which is classical, see e.g.~\cite{DBLP:conf/tlca/Atkey09} for an in-depth treatment -- between $\Tm{\UntypedTerms}$ and the set of closed untyped $\lambda$-terms modulo $\alpha$-renaming, i.e.\ syntax trees with binders, without $\beta$-conversion. Here are examples of encodings of the latter into the former (for the general definition, see \Cref{appendix:untyped-affine}):
\begin{align*}
    \lambda x.\,x\,x
    \quad\rightsquigarrow\quad&
    \colohide{\lambda(\ell : (\tyo \tto \tyo) \tto \tyo).\,\lambda(a : \tyo \tto \tyo \tto \tyo).}\,
    \\
    &\quad
    \colohide{\ell\ (}\lambda\colohide{(}x \colohide{\ : \tyo)}.\,\colohide{a}\ x\ x\colohide{)}
    \\
    (\lambda x.\,x\,x)\ (\lambda x.\,x\,x)
    \quad\rightsquigarrow\quad&
    \colohide{\lambda(\ell : (\tyo \tto \tyo) \tto \tyo).\,\lambda(a : \tyo \tto \tyo \tto \tyo).}\,
    \\
    &\quad
    \colohide{a\ (\ell}\ (\lambda\colohide{(}x \colohide{\ : \tyo)}.\,\colohide{a}\ x\ x)\colohide{)}\ \colohide{(\ell}\ (\lambda\colohide{(}x \colohide{\ : \tyo)}.\,\colohide{a}\ x\ x)\colohide{)}
    \\
    \lambda f.\,(\lambda x.\,x\,x)\,(\lambda x.\,f\,(x\,x))
    \quad\rightsquigarrow\quad&
    \colohide{\lambda(\ell : (\tyo \tto \tyo) \tto \tyo).\,\lambda(a : \tyo \tto \tyo \tto \tyo).}\,
    \\
    &\quad
    \colohide{\ell\ (}\lambda\colohide{(}f \colohide{\ : \tyo)}.\,\colohide{a\ (\ell}\ 
    (\lambda\colohide{(}x \colohide{\ : \tyo)}.\,\colohide{a}\ x\ x)\colohide{)}\\
    &\quad
    \phantom{\colohide{\ell\ (}\lambda\colohide{(}f \colohide{\ : \tyo)}.\,\colohide{a\ (\ell}}\
    \colohide{(\ell}\ (\lambda\colohide{(}x \colohide{\ : \tyo)}.\, \colohide{a}\ f\ (\colohide{a}\ x\ x))\colohide{))}
\end{align*}
This can be seen as an extension of Church encodings to higher-order abstract syntax: indeed, the variable $\ell$ plays the role of a constructor and introduces a bound variable.

A closed untyped term is \emph{affine} if and only if every bound variable occurs at most once. We now give the outline of the proof, detailed in \Cref{appendix:untyped-affine}, that the encodings in $\Tm{\UntypedTerms}$ of closed untyped affine terms form a \salvatiReg{\FinSet} language.

Let $Q$ be the finite set $\{0,1,\infty\}\times\{\top,\bot\}$, where $\{0,1,\infty\}$ is seen as the additive monoid $\N$ truncated to $2 = 3 = \dots = \infty$. We consider the two set-theoretic functions $\fapp : Q \to (Q \tto Q)$ and $\fabs : (Q \tto Q) \to Q$ defined as
\[
    \fapp(k, b)(k', b')
    \ :=\ 
    (k + k',\; b \wedge b')
    \quand
    \fabs(g)
    \ :=\ 
    (g_1(0, \top),\; g_2(0, \top) \wedge (g_1(1, \top) \le 1))
\]
where $g_1 : Q \to \{0, 1, \infty\}$ and $g_2 : Q \to \{\top, \bot\}$ are the compositions of $g : Q \to Q$ with the two projections. We verify that, for any closed untyped term $t$, we have
\[
    \semtm{t}{Q}(\fabs)(\fapp)
    \quad=\quad
    (0, b)
    \qquad
    \text{where $b$ is $\top$ if and only if $t$ is affine.}
\]
Therefore, if $F$ is the subset of $\semty{\UntypedTerms}{Q}$ defined as
\[
    F
    \ :=\ 
    \{
        s \in ((Q \tto Q) \tto Q) \tto (Q \tto Q \tto Q) \tto Q
        \mid
        s(\fabs)(\fapp) = (0, \top)
    \}
\]
then the \salvatiReg{\FinSet} language $\Lang{F}$ of terms of type $\UntypedTerms$ is the language of affine terms.

To the best of our knowledge, this regularity result is original. It also strongly suggests that the notion of \salvatiReg{\FinSet} language, when applied to the type $\UntypedTerms$ of syntax trees with binders, differs from the recognizability of these syntax trees by the nominal tree automata of~\cite[\S3.1]{vanheerdt_et_al:LIPIcs:2019:11434}. Indeed, nominal automata cannot recognize the language of data words whose letters are all different~\cite[Proof of Lemma 5.4]{DBLP:journals/corr/BojanczykKL14}.
\end{example}

\begin{remark}
    \label{rmk:CCC-functor-lambda-terms}
    Let $\C$ and $\D$ be CCCs and $G : \C \to \D$ be a CCC functor. By the universal property of $\Lam$, see \Cref{diag:lam-universal-property}, for every object $c$ of $\C$, the following diagram commutes:
\[
\begin{tikzcd}[ampersand replacement=\&]
	\Lam \\
	1 \& \C \& \D
	\arrow["\tyo", from=2-1, to=1-1]
	\arrow["c"', from=2-1, to=2-2]
	\arrow["G"', from=2-2, to=2-3]
	\arrow["{\semtm{-}{c}}", from=1-1, to=2-2]
	\arrow["{\semtm{-}{G(c)}}", curve={height=-18pt}, from=1-1, to=2-3]
\end{tikzcd}
\]
In particular, this means that for every simple type $A$, the objects $\semty{A}{G(c)}$ and $G(\semty{A}{c})$ are equal, and that for every simply typed $\lambda$-terms $t$ and $t'$ of type $A$,
\[
    \text{if}\quad
    \semtm{t}{c}\ =\ \semtm{t'}{c}
    \qquad
    \text{then}\quad
    \semtm{t}{G(c)}\ =\ \semtm{t'}{G(c)}
    \ .
\]
This means that interpreting the simply typed $\lambda$-calculus at the object $c$ will always be at least as fine as interpreting it at~$G(c)$. In particular, this entails that $\Rec[G(c)]{A} \subseteq \Rec[c]{A}$. If $G$ is moreover faithful, we then have that $\Rec[G(c)]{A} = \Rec[c]{A}$ for any object $c$ of $\C$. Therefore, all languages which are \salvatiReg{\C} are \salvatiReg{\D}.
\end{remark}

\paragraph*{Recognizable languages of~\pdflambda-terms, syntactically}

A syntactic approach to recognition is described in \cite{DBLP:conf/lics/HillebrandK96}. This syntactic approach uses the type substitution, also called cast, whose definition we now recall.

\begin{definition}
    If~$A$ and~$B$ are simple types, we define a simple type~$A[B] = A\{\tyo := B\}$ by replacing every occurence of~$\tyo$ in $A$ by~$B$. We extend this to $\lambda$-terms by induction:
        \begin{align*}
            x[B]
        \ &:=\ 
        x
        &
        (\lambda (x : A).\,t)[B]
        \ &:=\ 
        \lambda (x : A[B]).\,t[B]
        &
        (t\ u)[B]
        \ &:=\ 
        t[B]\ u[B]
        \\
        &&
        \langle t, u\rangle[B]
        \ &:=\ \langle t[B], u[B]\rangle
        &
        t_i[B]
        \ &:=\ 
        (t[B])_i
        \end{align*}
\end{definition}
\begin{claim}
    For every simple types $A$ and $B$ and $t \in \Tm{A}$, we have $t[B] \in \Tm{A[B]}$.
\end{claim}
\begin{remark}
    A more categorical way to understand casting is to see it as the unique CCC functor~$(-)[B] : \Lam \to \Lam$ such that the following diagram commutes:
    \[
\begin{tikzcd}[ampersand replacement=\&]
	\Lam \\
	1 \& \Lam
	\arrow["B"', from=2-1, to=2-2]
	\arrow["\tyo", from=2-1, to=1-1]
	\arrow["{(-)[B]}", dashed, from=1-1, to=2-2]
\end{tikzcd}
\]
As such, it is the semantic bracket functor $\semtm{-}{B} : \Lam \to \Lam$.
\end{remark}

Finally, we recall the encoding of Booleans into the simply typed $\lambda$-calculus.

\begin{definition}
    \label{def:bool}\AP
    Let~$\intro*\Bool$ be the simple type~$\tyo^2 \tto \tyo$. Its two inhabitants are the $\lambda$-terms
    \[
        \intro*\true
        \ :=\ 
        \lambda(x : \tyo^2).\,x_1
        \qquand
        \intro*\false
        \ :=\ 
        \lambda(x : \tyo^2).\,x_2
        \ .
    \]
\end{definition}

The following definition naturally generalizes the one given in \cite{DBLP:conf/lics/HillebrandK96} for $A=\Church{n}$.

\begin{definition}
    \label{def:reg-syn}\AP
    For any simple type $A$, a language $L \subseteq \Tm{A}$ is \intro*\hkReg if there exists a simple type $B$ and a $\lambda$-term $r$ of type $A[B] \tto \Bool$ such that
    \[
        L
        \ =\ 
        \{t \in \Tm{A} \mid r\ t[B] \eqbe \true\}
        \ .
    \]
\end{definition}

\begin{theorem}[Hillebrand \& Kanellakis, {{\cite[Theorem~3.4]{DBLP:conf/lics/HillebrandK96}}}] 
    \label{thm:hk-paper}
    A language of $\lambda$-terms of type $\Church{n}$ is \hkReg if and only if the associated language of finite words by the "Church encoding" is regular in the usual sense.
\end{theorem}

\begin{example}
    The "Church encodings" of words in $\{1,2\}^*$ with an even number of 1s and an odd number of 2s is \hkReg. Indeed, we consider the following $\lambda$-terms
    \begin{align*}
        \mathsf{and}
        \ &:=\ 
        \lambda(p:\Bool\times\Bool).\,\lambda(x:\tyo\times\tyo).\,
        p_1\ \langle p_2\ x, x_2 \rangle
        \\
        \mathsf{id}
        \ &:=\ 
        \lambda(b:\Bool).\,b
        \\
        \mathsf{not}
        \ &:=\ 
        \lambda(b:\Bool).\,\lambda(x:\tyo\times\tyo).\,
        b\ \langle x_2, x_1\rangle
    \end{align*}
    and choose, as in \Cref{def:reg-syn}, the type $B$ to be $\Bool$ and the simply typed $\lambda$-term $r$ to be
\[
    \lambda(w:\Church2 [\Bool]).\,
    \mathsf{and}\ 
    \langle w\ \mathsf{not}\ \mathsf{id}\ \true,\
    w\ \mathsf{id}\ \mathsf{not}\ \false\rangle
    \quad:\quad
    \Church2 [\Bool] \tto \Bool
    \ .
\]
Just as in \Cref{ex:odd-salvati}, this term can be seen as running two DFAs, both having two states, over the encoding of an input word.
\end{example}

When restricted to types $\Church{n}$ for any natural number $n$, the notion of \salvatiReg{\FinSet} and \hkReg languages coincide, as they both boil down to the usual notion of regular language of finite words, as seen in \Cref{ex:church-encoding} and \Cref{thm:hk-paper} respectively. One of the contributions of the present paper is to show that these two notions coincide at every simple type.

\section{Syntactic recognition implies semantic recognition}
\label{sec:semantic-evaluation}

In this section, we state \Cref{thm:semantic-evaluation} and prove it by extending the semantic evaluation argument, described by Hillebrand and Kanellakis in their proof of \Cref{thm:hk-paper} for the case of languages of words, to the more general case of languages of any type.

\begin{definition}
    \label{def:non-degenerate}\AP
    A CCC $\C$ is said to be ""non-degenerate"" if there exist two objects $c',c$ of~$\C$ such that there exist two distinct morphisms $f,g\colon c' \to c$.
\end{definition}

\begin{theorem}
    \label{thm:semantic-evaluation}
    If $\C$ is a non-degenerate CCC, then any language of $\lambda$-terms which is \hkReg is \salvatiReg{\C}.
\end{theorem}

\begin{proof}
    We observe first that the non-degeneracy assumption means that the two projections $\pi_1,\pi_2\colon c \times c \to c$ are not equal, since they yield different results when pre-composed with the morphism $c' \to c \times c$ obtained by pairing $f$ and $g$.

    Let~$A$ be a simple type and~$L$ be a language of~$\lambda$-terms of type~$A$ which is \hkReg. There exists a simple type~$B$ together with a~$\lambda$-term~$r \in \Tm{A[B] \tto \Bool}$ such that
\[
    L
    \quad=\quad
    \{t \in \Tm{A} \mid r\ t[B] \eqbe \true\}
    \ .
\]
In order to show that $L$ belongs to $\Rec[\semty{B}{c}]{A}$, we work with the interpretation of the simply typed~$\lambda$-calculus at the object~$\semty{B}{c}$ of~$\C$. By the universal property of the CCC~$\Lam$, the following diagram commutes
\[
\begin{tikzcd}[ampersand replacement=\&]
    \Lam \\
    1 \& \Lam \& \C
    \arrow["B", from=2-1, to=2-2]
    \arrow["{\semtm{-}{\semty{B}{c}}}", curve={height=-18pt}, from=1-1, to=2-3]
    \arrow["\tyo", from=2-1, to=1-1]
    \arrow["{(-)[B]}", from=1-1, to=2-2]
    \arrow["{\semtm{-}{c}}", from=2-2, to=2-3]
\end{tikzcd}
\ .
\]

More concretely, this states that, for every simply typed~$\lambda$-term~$t$, the two morphisms~$\semtm{t[B]}{c}$ and $\semtm{t}{\semty{B}{c}}$ are equal. By viewing $r$ as a morphism from $A[B]$ to $\Bool$ in the category~$\Lam$, the compositionality of the semantic interpretation gives us that
\[
    \semtm{r\ t[B]}{c}
    \ =\ 
    \semtm{r}{c} \circ \semtm{t[B]}{c}
    \ =\ 
    \semtm{r}{c} \circ \semtm{t}{\semty{B}{c}}
    \ .
\]
By non-degeneracy, the semantic interpretation $\semtm{-}{c} : \Tm{\Bool} \to \C(1, \semty{\Bool}{c})$ is injective: indeed, $\semtm{\true}{c} = \pi_1 \neq \pi_2 = \semtm{\false}{c}$. Therefore, we get the following equivalences
\[
    t \in L
    \ \iff\ 
    r\ t[B] \eqbe \true
    \ \iff\ 
    \semtm{r\ t[B]}{c} = \semtm{\true}{c}
    \ \iff\ 
    \semtm{t}{\semty{B}{c}} \in F
\]
where~$F$ is the subset of~$\C(1, \semty{A}{\semty{B}{c}})$ defined as $\{q \in \C(1, \semty{A}{\semty{B}{c}}) \mid \semtm{r}{c} \circ q = \semtm{\true}{c}\}$. This shows that $L$ belongs to $\Rec[\semty{B}{c}]{A}$, hence that it is a \salvatiReg{\C} language.
\end{proof}

\section{Logical relations and the squeezing construction}
\label{sec:relations-squeezing}

\paragraph*{Sconing in a nutshell}

In this paragraph, we recall the construction of a "CCC of logical relations" from one CCC to another. We first recall the construction of logical predicates, also called sconing, which will be general enough to give logical relations as a special case. This method is well-known, see for instance~\cite{DBLP:conf/csl/MitchellS92} for an introductory account.

\begin{definition}
    \label{def:ccc-logical-predicates}\AP
    Let $\C$ be a CCC. The category of logical predicates over $\C$, that we denote by~$\intro*\Pred{\C}$, is defined as follows:
    \begin{itemize}
        \item its objects are the pairs $(c, S)$ of an object $c$ of $\C$ together with a subset $S \subseteq \C(1, c)$,
        \item its morphisms from $(c, S)$ to $(c', S')$ are the morphisms $f : c \to c'$ of $\C$ such that $f \circ (-)$ restricts to a set-theoretic function $S \to S'$.
    \end{itemize}
\end{definition}
\begin{claim}
    This category $\Pred{\C}$ is a CCC, with exponentiation given by
    \[
        (c, S) \tto (c', S')
        \quad=\quad
        (c \tto c', \{f \in \C(1, c \tto c') \mid \forall s \in S, \ev_{c, c'} \circ \langle f, s\rangle \in S'\})
    \]
    The forgetful functor $(c, S) \mapsto c$ is a CCC functor.
\end{claim}
\begin{claimproof}
    See~\cite[p.~5]{DBLP:conf/csl/MitchellS92}, where the notation $\tilde{\C}$ is used for the category $\Pred{\C}$.
\end{claimproof}

\begin{definition}
    \label{def:ccc-logical-relations}\AP
    Let $\C_1$ and $\C_2$ be two CCCs. The ""CCC of logical relations"" from~$\C_1$ to~$\C_2$ is the CCC $\Pred{\C_1 \times \C_2}$, which admits a CCC functor $\Pred{\C_1 \times \C_2} \to \C_1 \times \C_2$.
\end{definition}

\begin{remark}
    We have defined the "CCC of logical relations" in terms of the logical predicate construction $\Pred{-}$ of \Cref{def:ccc-logical-predicates}. More concretely, this construction gives a category that can be described in the following way:
    \begin{itemize}
        \item its objects are triples $(c_1, c_2, \Vdash)$ where $c_i$ is an object of $\C_i$ for $i = 1, 2$ and $\Vdash$ is a subset of $\C_1(1, c_1) \times \C_2(1, c_2)$, and is thus a relation between the points of $c_1$ and of $c_2$,
        \item its morphisms from $(c_1, c_2, \Vdash)$ to $(c_1', c_2', \Vdash')$ are pairs $(f_1, f_2) \in \C_1(c_1, c_1') \times \C_2(c_2, c_2')$ such that for every pair $(x_1, x_2) \in \C_1(1, c_1) \times \C_2(1, c_2)$,
        \[
            \text{if}\ 
            x_1 \Vdash x_2
            \ ,
            \quad
            \text{then}\ 
            f_1 \circ x_1 \Vdash' f_2 \circ x_2
            \ .
        \]
    \end{itemize}
    For the proof that this category is a CCC, see~\cite[Proposition~4.3]{DBLP:conf/csl/MitchellS92}.
\end{remark}

\begin{remark}
    The "CCC of logical relations" from $\C_1$ to $\C_2$ comes with two projections to~$\C_1$ and to~$\C_2$ which are CCC functors. By \Cref{rmk:CCC-functor-lambda-terms}, we get that for any relation $(c_1, c_2, \Vdash)$,
\[
    \semty{A}{(c_1, c_2, \Vdash)}
    \quad=\quad
    (\semty{A}{c_1}, \semty{A}{c_2}, \Vdash^A)
    \quad\text{for some $\Vdash^A\ \subseteq\ \C_1(1, \semty{A}{c_1}) \times \C_2(1, \semty{A}{c_2})$.}
\]
The interpretation of a $\lambda$-term $t \in \Tm{A}$ at an object $(c_1, c_2, \Vdash)$ is a morphism of the form
    \[
        (\semtm{t}{c_1}, \semtm{t}{c_2})
        \quad:\quad
        1
        \ \longto\ 
        \semty{A}{(c_1, c_2, \Vdash)}
        \qquad
        \text{which means that}
        \qquad
        (\semtm{t}{c_1}, \semtm{t}{c_2}) \in \Vdash^A
        \ .
    \]
    This is the ""fundamental lemma of logical relations"", see e.g. \cite[Lemma~4.5.3]{DBLP:books/daglib/0093287-amadio-curien}.
\end{remark}

\begin{remark}
    At this stage, the categories $\C_1$ and $\C_2$ play a symmetric role. However, this will not always be the case in the rest of the paper, and we therefore say "CCC of logical relations" from $\C_1$ to $\C_2$ to emphasize the order.
\end{remark}

\paragraph*{The squeezing construction}

We describe here a construction, which we call squeezing, which produces a CCC $\Sqz{\C}$ from a CCC $\C$ equipped with an additional structure that we call a "squeezing structure". Intuitively, the objects of $\Sqz{\C}$ are objects of $\C$ coming with bounds induced by the structure, inspired by the squeeze theorem of calculus. This construction can be seen as the proof-irrelevant counterpart to the twisted gluing construction described in \cite[Definition~5]{DBLP:conf/ctcs/AltenkirchHS95}. The notion of "squeezing structure" that is used is related to the hypotheses of \cite[Lemma~6]{DBLP:journals/mscs/Fiore22}; this pattern also occurs in the older proof theory tradition, see for instance \cite[\S8.A]{girard:hal-01322183}.

Throughout this paragraph, we fix any CCC $\C$. We recall that a wide subcategory is a subcategory containing all objects, and can hence be seen as a predicate on morphisms, closed under finite compositions.

\begin{definition}
    \label{def:squeezing-structure}\AP
    A ""squeezing structure"" on $\C$ is the data of
    \begin{itemize}
        \item two wide subcategories~$\Cleft$ and~$\Cright$ of~$\C$ with associated notations~$\toleft$ and~$\toright$ for morphisms, which are stable under finite cartesian products and such that for all~$u : c_l \toleft c_l'$ and~$v : c_r \toright c_r'$,
        \[
            v \To u
            \ :\ 
            c_r' \To c_l \longtoleft c_r \To c_l'
            \qquand
            u \To v
            \ :\ 
            c_l' \To c_r \longtoright c_l \To c_r'
            \ .
        \]
        \item for every object~$c$ of~$\C$, two objects~$L_c$ and~$R_c$ of~$\C$ such that there exists morphisms:
    \begin{equation}
        \label{eq:left-right-morphisms}
    \begin{matrix}
        L_1 \longtoleft 1
        \quad&\quad
        L_{c \times c'} \longtoleft L_c \times L_{c'}
        \quad&\quad
        L_{c \tto c'} \longtoleft R_c \tto L_{c'}
        \\
        1 \longtoright R_1
        \quad&\quad
        R_c \times R_{c'} \longtoright R_{c \times c'}    
        \quad&\quad
        L_c \tto R_{c'} \longtoright R_{c \tto c'}
        \ .
    \end{matrix}
    \end{equation}
    \end{itemize}
\end{definition}

\begin{remark}
    As we work in a proof-irrelevant setting, we are merely interested in the existence of these morphisms. Nonetheless, knowing that they belong to~$\Cleft$ or~$\Cright$ gets us back some precious information, as we will see in \Cref{lem:sqz-partial-surj} and \Cref{sec:fin-encoding}.
\end{remark}

\begin{definition}
    \label{def:sqz-cat}\AP
    For a "squeezing structure" on $\C$, the category~$\intro*\Sqz{\C}$ is the full subcategory of~$\C$ whose objects are the objects~$c$ of~$\C$ for which there exist both a left morphism $L_c \toleft c$ and a right morphism $c \toright R_c$.
\end{definition}

Notice that we write $\Sqz{\C}$ even though this construction depends both on the CCC $\C$ and on a "squeezing structure" on $\C$.

\begin{theorem}
    \label{thm:squeezing-ccc}
    For a "squeezing structure" on $\C$, the category $\Sqz{\C}$ is a sub-CCC of~$\C$.
\end{theorem}

The proof is in \Cref{paragraph:proof-squeezing-ccc}.

\paragraph*{Partial surjections}

Logical relations which are partial surjections, i.e.\ both functional and surjective relations, can be a useful tool to obtain partial equivalence relations and to prove semantic results, see~\cite[\S3]{Bucciarelli96},~\cite[\S1.4.2]{DBLP:journals/tcs/Ehrhard12} and~\cite[Theorem~A]{entics:12280}. We now show that the squeezing construction can be applied to get partial surjections for free.

\begin{lemma}
    \label{lem:sqz-partial-surj}
    Let $\C_1$ and $\C_2$ be two CCCs and $\R$ be the "CCC of logical relations" from~$\C_1$ to~$\C_2$, whose objects are triples containing relations. Suppose that we are given a "squeezing structure" on $\R$ such that
    \begin{itemize}
        \item the relations in the objects $L_c$ are surjective,
        \item the relations in the objects $R_c$ are functional,
        \item the morphisms $(u_1, u_2)$ in $\Rleft$ are such that $\C_2(1, u_2)$ is a surjective function,
        \item the morphisms $(v_1, v_2)$ in $\Rright$ are such that $\C_2(1, v_2)$ is an injective function,
    \end{itemize}
    where, for $u \in \C(a, b)$, the function $\C(1, u) : \C(1, a) \to \C(1, b)$ is the composition $u \circ (-)$.

    Then, the relation of any object belonging to $\Sqz{\R}$ is a partial surjection.
\end{lemma}

The proof is in \Cref{appendix:proof-sqz-partial-surj}. We end this section by giving a definition which will appear in "squeezing structures" in \Cref{prop:squeezing-lam-finord} and in the proof of \Cref{prop:ccc-partial-surj}.

\begin{definition}
    \label{def:right-identity}\AP
    Let $\R$ be the "CCC of logical relations" from $\C_1$ to $\C_2$. We say that a morphism $(f_1, f_2) : (c_1, c_2, \Vdash) \to (c_1', c_2', \Vdash')$ is a ""target-identity"" if~$c_2$ and~$c_2'$ are the same object and if~$f_2$ is the identity morphism.
\end{definition}

\begin{remark}
    \label{rmk:right-identity-stability}
    We remark that the "target-identities" form a wide subcategory and are stable under products and exponentiation.
\end{remark}

\section{From well-pointed locally finite CCCs to finite sets}
\label{sec:partial-surjections}

We now recall the definition of the class of CCCs $\C$ for which we will show
that \salvatiReg{\C} languages coincide with our other definitions of regular languages of $\lambda$-terms.
Recall that in a CCC $\C$, a \emph{point} of an object~$c$ is a morphism~$1 \to c$ from the terminal object to~$c$.

\begin{definition}
  A~CCC is said to be:
  \begin{itemize}
  \item ""well-pointed"" if every morphism is determined by its action on points of its domain;
  \item ""locally finite"" if all its hom-sets are finite sets.
  \end{itemize}
\end{definition}

We start by introducing the following constructions, which will help us to use partial surjections.

\begin{definition}
    \label{def:ge-le-1}\AP
    Let $\C$ be a category with a terminal object. We define the following full subcategories of $\C$:
    \begin{itemize}
        \item $\Cge$ containing the objects $c$ that have at least one point; we call these objects ""inhabited"",
        \item $\Cle$ containing the objects $c$ that have at most one point,
        \item $\Ceq$ containing the objects $c$ that have exactly one point.
    \end{itemize}
\end{definition}

When instantiated to the CCC $\Lam$ of simple types and $\lambda$-terms, the notion of "inhabited object" coincides with the usual notion of inhabited simple type.

\begin{proposition}
    \label{prop:ge-ccc}
    If $\C$ is a CCC, then the category $\Cge$ is a sub-CCC of $\C$.
\end{proposition}

The proof is in \Cref{paragraph:proof-ge-ccc}.

\begin{proposition}
    \label{prop:le-eq-ccc}
    If $\E$ is a "well-pointed CCC", then the category $\Ele$ is a sub-CCC of $\E$. Moreover, the category $\Eeq$ is equivalent to the terminal category.
\end{proposition}

The proof is in \Cref{paragraph:proof-le-eq-ccc}. We now prove the interesting fact that "inhabited objects" characterize the recognized languages of $\lambda$-terms in a "well-pointed CCC".

\begin{proposition}
    \label{prop:well-pointed-ge-salvati}
    If $\E$ is a "well-pointed CCC", then a language of $\lambda$-terms is \salvatiReg{\E} if and only if it is \salvatiReg{\Ege}.
\end{proposition}

\begin{proof}
  Let $\E$ be a "well-pointed CCC". By \Cref{prop:ge-ccc}, $\Ege$ is a sub-CCC of $\E$ so any language which is \salvatiReg{\Ege} is \salvatiReg{\E}, as explained in \Cref{rmk:CCC-functor-lambda-terms}.

    We now show that the only language recognized by $\Ele$ is the empty and full languages. Let $A$ be a simple type and $c$ be an object of $\Ele$. As $\E$ is well-pointed and by \Cref{prop:le-eq-ccc}, we know that $\semty{A}{c}$ belongs to $\Ele$, which means that $\C(1, \semty{A}{c})$ is empty or a singleton, so~$\Rec[c]{A}$ contains at most the empty and full languages, which are the same when $A$ is not inhabited.

    The empty and full languages are recognized by any CCC, so in particular by $\Ege$. Therefore, all \salvatiReg{\E} languages are \salvatiReg{\Ege}.
\end{proof}

\begin{example}
    \begin{itemize}
        \item The category $\catge\FinSet$ is the CCC of \emph{non-empty} finite sets.
        \item An implicative semilattice is a meet-semilattice such that each meet operation has an upper adjoint. Implicative semilattices are CCCs, however, they are "degenerate" and so never distinguish different $\lambda$-terms of the same type.

        Another way to understand this fact is to remark that their full subcategory of "inhabited objects" is the terminal category.
    \end{itemize}
\end{example}

\noindent
Next, we introduce the basic partial relations at which we will interpret the $\lambda$-calculus.

\begin{definition}
    \label{def:well-pointed-bij}
    Let $\E$ be a "well-pointed" "locally finite" CCC and $\R$ be the "CCC of logical relations" from $\FinSet$ to $\E$. For any object $e$ of $\E$, we consider the triple
    \[
        T_e
        \quad:=\quad
        (\E(1, e), e, \sim_e)
        \qquad
        \text{where $\sim_e$ is the identity relation of $\E(1, e)$.}
    \]
    which extends to a functor $T : \E \to \R$.
\end{definition}

We now prove a converse to \Cref{lem:sqz-partial-surj} in the present case.

\begin{proposition}
    \label{prop:ccc-partial-surj}
    Let $\E$ be a "well-pointed" "locally finite" CCC whose objects are all "inhabited" and $\R$ be the "CCC of logical relations" from $\FinSet$ to $\E$. Then, the full subcategory of partial surjections is a sub-CCC of $\R$.
\end{proposition}

We provide two proofs of this proposition. The first one constructs the CCC by hand, see \Cref{paragraph:first-proof-ccc-partial-surj}. The other one uses the squeezing technique, see \Cref{paragraph:second-proof-ccc-partial-surj}.

\begin{theorem}[{claimed in \cite[Lemma~20]{DBLP:books/hal/Salvati15}}]
    \label{thm:salvati-points-finset}
    For every "well-pointed" "locally finite" CCC $\E$, any \salvatiReg{\E} language is~\salvatiReg{\FinSet}.
\end{theorem}

\begin{proof}
    Let $A$ be a simple type and $L$ be a language of $\lambda$-terms of type $A$ which is \salvatiReg{\E}. By \Cref{prop:well-pointed-ge-salvati}, it is \salvatiReg{\Ege}. Let $e$ be an object of $\Ege$ such that $L \in \Rec[e]{A}$. We consider the object $T_e = (\E(1, e), e, \sim_e)$ from \Cref{def:well-pointed-bij}, whose relation is a partial surjection. The object $\semty{A}{T_e}$ is of the form $(\semty{A}{\E(1, e)}, \semty{A}{e}, \sim_e^A)$, where the relation~$\sim_e^A$ is a partial surjection, as explained in \Cref{rmk:CCC-functor-lambda-terms}.

    Let $F$ be a subset of $\E(1, \semty{A}{e})$ such that $L$ is $\Lang{F}$. We consider the subset $F'$ of $\semty{A}{\E(1, e)}$ defined as the inverse image
    \[
        F'
        \quad:=\quad
        \{q \in \semty{A}{\E(1, e)} \mid \exists q' \in F \text{ s.t. } q \sim_e^A q'\}
    \]
    By the "fundamental lemma of logical relations", for $\lambda$-term $t$ of type $A$, we have
    \[
        \semtm{t}{\E(1, e)} \sim_e^A \semtm{t}{e}
        \ .
    \]
    which proves that $\Lang{F} \subseteq \Lang{F'}$. Moreover, as $\sim_e^A$ is a functional relation, we get the converse inclusion. This proves that $L$ is \salvatiReg{\FinSet}.
\end{proof}

\section{From finite sets to \pdflambda-terms}
\label{sec:fin-encoding}

In this section, we apply the squeezing construction of \Cref{def:sqz-cat} on a "CCC of logical relations" to show that every \salvatiReg{\FinSet} language is \hkReg, through an encoding of finite sets into the simply typed $\lambda$-calculus. To achieve that, we need to change slightly of setting, by moving from finite sets to finite ordinals. This will make it possible to define the functor $\Fin{-}$ without ambiguity.

Therefore, we consider the category $\FinOrd$ whose objects are natural numbers and whose morphisms are the set-theoretic maps between the associated finite cardinals~$\intro*\card{n} := \{1, \dots, n\}$. The inclusion of $\FinOrd$ in $\FinSet$ is a fully faithful functor that is essentially surjective, henceforth we get an equivalence between $\FinOrd$ and $\FinSet$ using the axiom of choice. In particular, $\FinOrd$ is a CCC that recognizes the same languages as $\FinSet$.

\paragraph*{The encoding of finite sets and its squeezing structure}

We take~$\R$ to be the "CCC of logical relations" from $\Lam$ to $\FinOrd$. We recall that the objects of $\R$ are therefore triples~$R = (B, n, \Vdash)$ where~$B$ is a simple type,~$n$ is a natural number and~$\Vdash$ is a subset of the product~$\Tm{B} \times \card{n}$. A morphism of~$\R$ is a pair of morphisms of~$\Lam$ and~$\FinOrd$ which respect the relations.

\begin{definition}
    \label{def:fin}
    We define the functor $\intro*\Fin{{-}} : \FinOrd \to \Lam$ as $\Fin{n} := \tyo^n \tto \tyo$ and, for every $f : n \to n'$, the $\lambda$-term $\Fin{f} : \Fin{n} \to \Fin{n'}$ is
    \[
        \lambda (p : \Fin{n}).\,\lambda (x : \tyo^{n'}).\,p\ \langle x_{f(1)},\,\dots,\,x_{f(n)}\rangle
        \ .
    \]
\end{definition}

\begin{remark}
    As~$\FinOrd$ is equivalent to the free cocartesian category and~$\op{\Lam}$ is cocartesian, we get a functor $n \mapsto \tyo^n$. The composition of the two functors
\[
\begin{tikzcd}[ampersand replacement=\&]
	\FinOrd \&\& {\op{\Lam}} \&\& \Lam
	\arrow["{n \mapsto \tyo^n}", from=1-1, to=1-3]
	\arrow["{(-) \tto \tyo}", from=1-3, to=1-5]
\end{tikzcd}
\]
is precisely the functor $\Fin{{-}}$.
\end{remark}

Our goal is now to exhibit a "squeezing structure" on $\R$ in order to show \Cref{thm:hk-squeezing}. We consider the "target-identities" for the two wide subcategories of the structure. We now define the following family of objects.

\begin{definition}
    \label{def:bij-logrel}
    For any natural number $n$, we define the object $\Bij{n}$ as
\[
    \intro*\Bij{n}
    \quad:=\quad
    \basebij{n}
\]
where~$\logrel{n}$ is the bijection between the sets~$\Tm{\Fin{n}}$ and~$\card{n}$ defined as
\[
    \intro*\logrel{n}
    \quad
    :=
    \quad
    \{(\pi_i, i) : 1 \le i \le n\}
    \qquad\text{with~$\pi_i$ the~$\lambda$-term~$\lambda(x : \tyo^n).\,x_i$}
    \ .
\]
This assignment extends to a functor $\Bij{(-)} : \FinOrd \to \R$.
\end{definition}

\begin{proposition}
    \label{prop:squeezing-lam-finord}
    There is a "squeezing structure" on the CCC $\R$ such that:
    \begin{itemize}
        \item the left and right morphisms are the "target-identities" of~$\R$,
        \item for any object $c = (B, n, \Vdash)$ of $\R$, the objects $L_c$ and $R_c$ are both equal to $\Bij{n}$.
    \end{itemize}
\end{proposition}

The proof is in \Cref{paragraph:proof-squeezing-lam-finord}. \Cref{prop:squeezing-lam-finord} shows that we have a sub-CCC~$\Sqz{\C}$ of the "CCC of logical relations" from~$\Lam$ to~$\FinOrd$, whose objects are tuples~$(B, n, \Vdash)$ such that there exists~$\lambda$-terms~$u : \Fin{n} \to B$ and~$v : B \to \Fin{n}$ lifting to the two following "target-identities":
\[
\begin{tikzcd}[ampersand replacement=\&]
    	{\basebij{n}} \&\& {(B, n, \Vdash)}
    	\arrow["{(u, \Id_n)}", from=1-1, to=1-3]
    \end{tikzcd}
    \quand
\begin{tikzcd}[ampersand replacement=\&]
    	{(B, n, \Vdash)} \&\& {\basebij{n}}
    	\arrow["{(v, \Id_n)}", from=1-1, to=1-3]
    \end{tikzcd}
    \ .
\]

\paragraph*{Encoding recognizability by finite sets}

We have shown in \Cref{prop:squeezing-lam-finord} that we have a "squeezing structure" on the "CCC of logical relations" $\R$ from $\Lam$ to $\FinOrd$. We now show how to use this structure, culminating in the link established in \Cref{thm:hk-squeezing} between \salvatiReg{\FinOrd} and \hkReg languages.

\begin{theorem}
    \label{thm:hk-squeezing}
    If a language is \salvatiReg{\FinSet}, then it is \hkReg.
\end{theorem}

\begin{proof}
Let $A$ be a simple type and $L \subseteq \Tm{A}$ be any \salvatiReg{\FinSet} language. There exists a finite set $Q$ and a subset~$F \subseteq \card{\semty{A}{n}}$ such that
\[
    L
    \ =\ 
    \{t \in \Tm{A} \mid \semtm{t}{n} \in F\}
    \ .
\]
We take $n$ to be the cardinality of $Q$ and note $\chi : \semty{A}{n} \to 2$ the characteristic function associated to the subset $F$. By applying the functor~$\Bij{(-)}$, we get a morphism of relations
\[
    \Bij{\chi} := (\Fin{\chi}\,,\, \chi)
    \quad:\quad
    \basebij{\semty{A}{n}}
    \ \longto\ 
    \basebij{2}
    \ .
\]
The interpretation $\semty{A}{\Bij{n}}$ is of the form $(A[\Fin{n}], \semty{A}{n}, \logrel[A]{n})$ as explained in \Cref{rmk:CCC-functor-lambda-terms}. As it is an object of $\Sqz{\R}$, it has a "target-identity" into $\Bij{\semty{A}{n}}$. By composing this morphism with $\Bij{\chi}$, we obtain a morphism
\[
    (r\,,\, \chi)
    \quad:\quad
    (A[\Fin{n}], \semty{A}{n}, \logrel[A]{n})
    \ \longto\ 
    \basebij{2}
    \ .
\]
By the "fundamental lemma of logical relations", we get that, for every~$\lambda$-term~$t$ of type~$A$,
\[
    t[\Fin{n}]
    \logrel[A]{n}
    \semtm{t}{n}
    \qquad\text{on which we apply $(r, \chi)$ to get}\qquad
    r\ t[\Fin{n}]
    \logrel{2}
    \chi(\semtm{t}{n})
\]
which states that $r\ t[\Fin{n}] \eqbe \true$ if and only if $\chi(\semtm{t}{n})$ is $1$. This proves that~$r$ recognizes the language $\Lang{F}$ given by~$F \subseteq \card{\semty{A}{n}}$, and so that $L$ is \hkReg.
\end{proof}

\section{Regular languages}
\label{sec:regular-languages}

In this section, we want to point out a few consequences of the equivalence previously proved through \Cref{thm:semantic-evaluation}, \Cref{thm:salvati-points-finset} and \Cref{thm:hk-squeezing}. Using these theorems, the following definition of "regular languages" is well-defined.

\begin{definition}
    \label{def:regular-language}\AP
    Let $A$ be a simple type. A ""regular language"" of $\lambda$-terms of type $A$ is a subset $L \subseteq \Tm{A}$ such that one of the following equivalent propositions holds:
    \begin{itemize}
        \item $L$ is \hkReg;
        \item $L$ is \salvatiReg{\C}, for some "non-degenerate", "well-pointed" and "locally finite" CCC $\C$;
        \item $L$ is \salvatiReg{\FinSet}.
    \end{itemize}
    We denote by $\intro*\Reg{A}$ the set of "regular languages" of $\lambda$-terms of type $A$.
\end{definition}

\begin{remark}
    Note that $\FinSet$ recognizes all the "regular languages" of $\lambda$-terms. In that sense, it plays the same role as the monoid
    \[
        M
        \ :=\ 
        \{f : \N \to \N \mid \exists N \in \N, \forall n \ge N, f(n) = n\}
        \quad\text{with the composition of functions}
    \]
    which recognizes all the regular languages of finite words as all finite monoids can be embedded into $M$. Such a monoid cannot be finite; however, $M$ is a \emph{locally finite monoid}, i.e.\ all its finitely generated submonoids are finite (this is a standard notion, see e.g.~\cite[\S{}V.5]{grilletsemigroups}).

    In the case of finite words, recognizability by finite monoids and locally finite monoids are equivalent when the alphabet is finite. In the case of $\lambda$-terms however, finite CCCs are all "degenerate" whereas the "locally finite" case yields "regular languages" of $\lambda$-terms, with some additional conditions.
\end{remark}

\begin{proposition}
    \label{prop:regular-languages-product}
    The set $\Reg{A}$ of "regular languages" of $\lambda$-terms of some simple type $A$ is a Boolean algebra.
\end{proposition}

This fact boils down to stability by union or intersection. It is proved in \cite[Theorem~8]{DBLP:conf/wollic/Salvati09} using intersection types and in \cite[Proposition~2.5]{entics:12280} using logical relations. We provide another proof, showing that it is a direct consequence of our results.

\begin{proof}
    Using any of the three conditions of \Cref{def:regular-language}, it is clear that "regular languages" are closed under complement.

    The product CCC $\FinSet \times \FinSet$ is "non-degenerate", "well-pointed" and "locally finite". It comes with two projections which are both CCC functors $\FinSet\times\FinSet\to\FinSet$. Let $Q$ and $Q'$ be two finite sets; we consider the object $(c, c')$ of~$\FinSet \times \FinSet$. As explained in \Cref{rmk:CCC-functor-lambda-terms}, for any simple type $A$, we get that $\Rec[Q]{A} \subseteq \Rec[(Q, Q')]{A}$ and $\Rec[Q']{A} \subseteq \Rec[(Q, Q')]{A}$.

    Moreover, $\Rec[(Q, Q')]{A}$ is a Boolean algebra. This shows that the intersection of a language in $\Rec[Q]{A}$ with another in $\Rec[Q']{A}$ can be taken in $\Rec[(Q, Q')]{A}$, so it is still a "regular language". Therefore, $\Reg{A}$ is a Boolean algebra.
\end{proof}

We now point out two other consequences of the equivalence in \Cref{def:regular-language}.
\begin{itemize}
    \item As stated in the introduction, Statman's finite completeness theorem tells us that singleton languages of $\lambda$-terms, taken modulo $\beta\eta$-conversion, are "regular languages". It has multiple proofs, see~\cite{statman-dowek-1992} for proof directly in the finite standard model, \cite{DBLP:conf/wollic/Salvati09} using intersection types, \cite{DBLP:journals/iandc/KobeleS15} in the model of complete lattices and \cite{DBLP:journals/ipl/SrivathsanW12} using Böhm trees.

      These results are usually proved in one CCC. Using \Cref{thm:semantic-evaluation}, \Cref{thm:salvati-points-finset} and \Cref{thm:hk-squeezing}, we get that the singleton languages are recognized by any "non-degenerate", "well-pointed" and "locally finite" CCC and are also \hkReg.
    \item Some CCCs satisfying these three conditions are the coKleisli categories of a model of linear logic, see~\cite{mellies_2009}. In \cite{HOParity}, a notion of higher-order automaton is presented, which recognizes a language of $\lambda$-terms of a given simple type. The run-trees for these non-deterministic automata are defined using an intersection type system, which is an equivalent way of presenting the semantic interpretation~$\semtm{-}{}$ of the simply typed $\lambda$-calculus in the coKleisli category $\mathbf{ScottL}_{!}$ of the Scott model of linear logic. Using the equivalence proved in the present article, a language is recognized by a higher-order automaton, i.e. \salvatiReg{\mathbf{ScottL}_{!}}, if and only if it satisfies one of the conditions of \Cref{def:regular-language}.
\end{itemize}

\section{Conclusion and future perspectives}
\label{sec:conclusion}

In this article, we have shown that every "non-degenerate", "well-pointed" and
"locally finite" CCCs recognizes exactly Salvati's "regular languages" of
$\lambda$-terms~\cite{DBLP:conf/wollic/Salvati09}, and that those also coincide
with \hkReg languages. This is evidence for the robustness of this notion, and
therefore of the dual notion of profinite $\lambda$-term introduced in
\cite{entics:12280}.

Among the aforementioned conditions, non-degeneracy is needed to recognize
non-trivial languages, and local finiteness is clearly crucial: in the case of
finite words and trees, regularity is closely related to recognition by finitary
structures. What about well-pointedness? In other words, one question that
remains open is the following: is there a locally finite CCC that recognizes
languages of $\lambda$-terms that are not "regular", i.e.\ not recognizable by
$\FinSet$? For example, sequential algorithms famously form a locally finite CCC
which is \emph{not} well-pointed,
cf.~\cite[Chapter~14]{DBLP:books/daglib/0093287-amadio-curien}; we would like to
understand its recognition power.

As explained in \Cref{ex:church-encoding}, the "regular languages" of $\lambda$-terms of type $\Church{n}$ for some natural number $n$ are exactly the usual regular languages of the finite words associated by the "Church encoding". It is possible to encode words in other calculi, like the non-commutative affine $\lambda$-calculus in which case a syntactic approach analogous to \Cref{def:reg-syn} yields the star-free languages, see \cite{iatlc1}. Moreover, gluing techniques have been studied for other calculi, see \cite{DBLP:conf/tlca/Hasegawa99} for the linear case. One can therefore wonder whether it is possible to develop a semantic approach \textit{à la} Salvati, analogous to \Cref{def:reg-sem}, for other calculi.



\bibliography{biblio}

\begin{thebibliography}{10}

\bibitem{DBLP:conf/ctcs/AltenkirchHS95}
Thorsten Altenkirch, Martin Hofmann, and Thomas Streicher.
\newblock Categorical reconstruction of a reduction free normalization proof.
\newblock In David~H. Pitt, David~E. Rydeheard, and Peter~T. Johnstone,
  editors, {\em Category Theory and Computer Science, 6th International
  Conference, {CTCS} '95, Cambridge, UK, August 7-11, 1995, Proceedings},
  volume 953 of {\em Lecture Notes in Computer Science}, pages 182--199.
  Springer, 1995.
\newblock \href {https://doi.org/10.1007/3-540-60164-3_27}
  {\path{doi:10.1007/3-540-60164-3_27}}.

\bibitem{DBLP:books/daglib/0093287-amadio-curien}
Roberto~M. Amadio and Pierre{-}Louis Curien.
\newblock {\em Domains and {L}ambda-{C}alculi}, volume~46 of {\em Cambridge
  Tracts in Theoretical Computer Science}.
\newblock Cambridge University Press, 1998.
\newblock \href {https://doi.org/10.1017/CBO9780511983504}
  {\path{doi:10.1017/CBO9780511983504}}.

\bibitem{DBLP:conf/tlca/Atkey09}
Robert Atkey.
\newblock Syntax for free: Representing syntax with binding using
  parametricity.
\newblock In Pierre{-}Louis Curien, editor, {\em Typed Lambda Calculi and
  Applications, 9th International Conference, {TLCA} 2009, Brasilia, Brazil,
  July 1-3, 2009. Proceedings}, volume 5608 of {\em Lecture Notes in Computer
  Science}, pages 35--49. Springer, 2009.
\newblock \href {https://doi.org/10.1007/978-3-642-02273-9_5}
  {\path{doi:10.1007/978-3-642-02273-9_5}}.

\bibitem{DBLP:journals/corr/abs-2008-11635}
Mikołaj Bojańczyk.
\newblock Languages recognised by finite semigroups, and their generalisations
  to objects such as trees and graphs, with an emphasis on definability in
  monadic second-order logic.
\newblock {\em CoRR}, abs/2008.11635, 2020.
\newblock \href {https://arxiv.org/abs/2008.11635} {\path{arXiv:2008.11635}}.

\bibitem{DBLP:journals/corr/BojanczykKL14}
Mikołaj Bojańczyk, Bartek Klin, and Sławomir Lasota.
\newblock Automata theory in nominal sets.
\newblock {\em Logical Methods in Computer Science}, 10(3), 2014.
\newblock \href {https://doi.org/10.2168/LMCS-10(3:4)2014}
  {\path{doi:10.2168/LMCS-10(3:4)2014}}.

\bibitem{Bucciarelli96}
Antonio Bucciarelli.
\newblock Logical relations and lambda theories.
\newblock In {\em Advances in Theory and Formal Methods of Computing,
  proceedings of the 3rd Imperial College Workshop}, pages 37--48, 1996.

\bibitem{DBLP:journals/tcs/Ehrhard12}
Thomas Ehrhard.
\newblock The {Scott} model of linear logic is the extensional collapse of its
  relational model.
\newblock {\em Theoretical Computer Science}, 424:20--45, 2012.
\newblock \href {https://doi.org/10.1016/j.tcs.2011.11.027}
  {\path{doi:10.1016/j.tcs.2011.11.027}}.

\bibitem{DBLP:journals/mscs/Fiore22}
Marcelo Fiore.
\newblock Semantic analysis of normalisation by evaluation for typed lambda
  calculus.
\newblock {\em Mathematical Structures in Computer Science}, 32(8):1028--1065,
  2022.
\newblock \href {https://doi.org/10.1017/S0960129522000263}
  {\path{doi:10.1017/S0960129522000263}}.

\bibitem{girard:hal-01322183}
Jean-Yves Girard.
\newblock {\em The {Blind} {Spot}: {Lectures} on logic}.
\newblock European Mathematical Society, September 2011.
\newblock \href {https://doi.org/10.4171/088} {\path{doi:10.4171/088}}.

\bibitem{grellois:tel-01311150}
Charles Grellois.
\newblock {\em {Semantics of linear logic and higher-order model-checking}}.
\newblock PhD thesis, {Université Paris 7}, April 2016.
\newblock URL: \url{https://hal.science/tel-01311150}.

\bibitem{grilletsemigroups}
Pierre~A. Grillet.
\newblock {\em Semigroups. An introduction to the structure theory}.
\newblock Chapman \& Hall/CRC Pure and Applied Mathematics. Dekker, 1995.
\newblock \href {https://doi.org/10.4324/9780203739938}
  {\path{doi:10.4324/9780203739938}}.

\bibitem{DBLP:conf/tlca/Hasegawa99}
Masahito Hasegawa.
\newblock Logical predicates for intuitionistic linear type theories.
\newblock In Jean{-}Yves Girard, editor, {\em Typed Lambda Calculi and
  Applications, 4th International Conference, TLCA'99, L'Aquila, Italy, April
  7-9, 1999, Proceedings}, volume 1581 of {\em Lecture Notes in Computer
  Science}, pages 198--212. Springer, 1999.
\newblock \href {https://doi.org/10.1007/3-540-48959-2_15}
  {\path{doi:10.1007/3-540-48959-2_15}}.

\bibitem{vanheerdt_et_al:LIPIcs:2019:11434}
Gerco~van Heerdt, Tobias Kappé, Jurriaan Rot, Matteo Sammartino, and Alexandra
  Silva.
\newblock {Tree Automata as Algebras: Minimisation and Determinisation}.
\newblock In Markus Roggenbach and Ana Sokolova, editors, {\em 8th Conference
  on Algebra and Coalgebra in Computer Science (CALCO 2019)}, volume 139 of
  {\em Leibniz International Proceedings in Informatics (LIPIcs)}, pages
  6:1--6:22, Dagstuhl, Germany, 2019. Schloss Dagstuhl--Leibniz-Zentrum für
  Informatik.
\newblock \href {https://doi.org/10.4230/LIPIcs.CALCO.2019.6}
  {\path{doi:10.4230/LIPIcs.CALCO.2019.6}}.

\bibitem{DBLP:conf/lics/HillebrandK96}
Gerd~G. Hillebrand and Paris~C. Kanellakis.
\newblock On the expressive power of simply typed and let-polymorphic lambda
  calculi.
\newblock In {\em Proceedings, 11th Annual {IEEE} Symposium on Logic in
  Computer Science, New Brunswick, New Jersey, USA, July 27-30, 1996}, pages
  253--263. {IEEE} Computer Society, 1996.
\newblock \href {https://doi.org/10.1109/LICS.1996.561337}
  {\path{doi:10.1109/LICS.1996.561337}}.

\bibitem{DBLP:journals/jacm/Kobayashi13}
Naoki Kobayashi.
\newblock Model checking higher-order programs.
\newblock {\em Journal of the {ACM}}, 60(3):20:1--20:62, 2013.
\newblock \href {https://doi.org/10.1145/2487241.2487246}
  {\path{doi:10.1145/2487241.2487246}}.

\bibitem{DBLP:conf/ppdp/Kobayashi19}
Naoki Kobayashi.
\newblock 10 years of the higher-order model checking project (extended
  abstract).
\newblock In Ekaterina Komendantskaya, editor, {\em Proceedings of the 21st
  International Symposium on Principles and Practice of Programming Languages,
  {PPDP} 2019, Porto, Portugal, October 7-9, 2019}, pages 2:1--2:2. {ACM},
  2019.
\newblock \href {https://doi.org/10.1145/3354166.3354167}
  {\path{doi:10.1145/3354166.3354167}}.

\bibitem{DBLP:journals/iandc/KobeleS15}
Gregory~M. Kobele and Sylvain Salvati.
\newblock The {IO} and {OI} hierarchies revisited.
\newblock {\em Information and Computation}, 243:205--221, 2015.
\newblock \href {https://doi.org/10.1016/j.ic.2014.12.015}
  {\path{doi:10.1016/j.ic.2014.12.015}}.

\bibitem{mellies_2009}
Paul-André Melliès.
\newblock Categorical {Semantics} of {Linear} {Logic}.
\newblock In P.-L. Curien, H.~Herbelin, J.-L. Krivine, and P.-A. Melliès,
  editors, {\em Interactive models of computation and program behaviour},
  volume~27 of {\em Panoramas et Synthèses}. Société Mathématique de
  France, 2009.
\newblock URL:
  \url{https://smf.emath.fr/publications/semantique-categorielle-de-la-logique-lineaire}.

\bibitem{HOParity}
Paul-André Melliès.
\newblock Higher-order parity automata.
\newblock In {\em 2017 32nd {Annual} {ACM}/{IEEE} {Symposium} on {Logic} in
  {Computer} {Science} ({LICS})}, pages 1--12, Reykjavik, Iceland, June 2017.
  IEEE.
\newblock \href {https://doi.org/10.1109/LICS.2017.8005077}
  {\path{doi:10.1109/LICS.2017.8005077}}.

\bibitem{DBLP:conf/csl/MitchellS92}
John~C. Mitchell and Andre Scedrov.
\newblock Notes on sconing and relators.
\newblock In Egon B{ö}rger, Gerhard J{ä}ger, Hans~Kleine B{ü}ning, Simone
  Martini, and Michael~M. Richter, editors, {\em Computer Science Logic, 6th
  Workshop, {CSL} '92, San Miniato, Italy, September 28 - October 2, 1992,
  Selected Papers}, volume 702 of {\em Lecture Notes in Computer Science},
  pages 352--378. Springer, 1992.
\newblock \href {https://doi.org/10.1007/3-540-56992-8_21}
  {\path{doi:10.1007/3-540-56992-8_21}}.

\bibitem{titoPhD}
Lê Thành~Dũng Nguy\~{\^e}n.
\newblock {\em Implicit automata in linear logic and categorical transducer
  theory}.
\newblock PhD thesis, Université Paris XIII, 2021.
\newblock URL: \url{https://hal.science/tel-04132636}.

\bibitem{iatlc1}
Lê Thành D\~ung Nguy{\~{ê}}n and Cécilia Pradic.
\newblock Implicit automata in typed $\lambda$-calculi {I}: aperiodicity in a
  non-commutative logic.
\newblock In Artur Czumaj, Anuj Dawar, and Emanuela Merelli, editors, {\em 47th
  International Colloquium on Automata, Languages, and Programming, {ICALP}
  2020, July 8-11, 2020, Saarbr{\"{u}}cken, Germany (Virtual Conference)},
  volume 168 of {\em LIPIcs}, pages 135:1--135:20. Schloss Dagstuhl -
  Leibniz-Zentrum f{\"{u}}r Informatik, 2020.
\newblock \href {https://doi.org/10.4230/LIPIcs.ICALP.2020.135}
  {\path{doi:10.4230/LIPIcs.ICALP.2020.135}}.

\bibitem{DBLP:conf/lics/Ong06}
C.{-}H.~Luke Ong.
\newblock On model-checking trees generated by higher-order recursion schemes.
\newblock In {\em 21th {IEEE} Symposium on Logic in Computer Science {(LICS}
  2006), 12-15 August 2006, Seattle, WA, USA, Proceedings}, pages 81--90.
  {IEEE} Computer Society, 2006.
\newblock \href {https://doi.org/10.1109/LICS.2006.38}
  {\path{doi:10.1109/LICS.2006.38}}.

\bibitem{DBLP:conf/lics/Ong15}
C.{-}H.~Luke Ong.
\newblock Higher-order model checking: An overview.
\newblock In {\em 30th Annual {ACM/IEEE} Symposium on Logic in Computer
  Science, {LICS} 2015, Kyoto, Japan, July 6-10, 2015}, pages 1--15. {IEEE}
  Computer Society, 2015.
\newblock \href {https://doi.org/10.1109/LICS.2015.9}
  {\path{doi:10.1109/LICS.2015.9}}.

\bibitem{DBLP:journals/corr/abs-2206-08413}
Gordon~D. Plotkin.
\newblock Recursion does not always help.
\newblock In Fairouz Kamareddine, editor, {\em A Century since Principia's
  Substitution Bedazzled Haskell Curry. In Honour of Jonathan Seldin's 80th
  Anniversary}. College Publications, 2023.
\newblock \href {https://arxiv.org/abs/2206.08413} {\path{arXiv:2206.08413}}.

\bibitem{ronchidellarocca:LIPIcs:2018:9861}
Simona Ronchi~Della Rocca.
\newblock {Intersection Types and Denotational Semantics: An Extended Abstract
  (Invited Paper)}.
\newblock In Silvia Ghilezan, Herman Geuvers, and Jelena Ivetić, editors, {\em
  22nd International Conference on Types for Proofs and Programs (TYPES 2016)},
  volume~97 of {\em Leibniz International Proceedings in Informatics (LIPIcs)},
  pages 2:1--2:7, Dagstuhl, Germany, 2018. Schloss Dagstuhl--Leibniz-Zentrum
  für Informatik.
\newblock \href {https://doi.org/10.4230/LIPIcs.TYPES.2016.2}
  {\path{doi:10.4230/LIPIcs.TYPES.2016.2}}.

\bibitem{DBLP:conf/wollic/Salvati09}
Sylvain Salvati.
\newblock Recognizability in the simply typed lambda-calculus.
\newblock In Hiroakira Ono, Makoto Kanazawa, and Ruy J. G.~B. de~Queiroz,
  editors, {\em Logic, Language, Information and Computation, 16th
  International Workshop, WoLLIC 2009, Tokyo, Japan, June 21-24, 2009.
  Proceedings}, volume 5514 of {\em Lecture Notes in Computer Science}, pages
  48--60. Springer, 2009.
\newblock \href {https://doi.org/10.1007/978-3-642-02261-6_5}
  {\path{doi:10.1007/978-3-642-02261-6_5}}.

\bibitem{DBLP:books/hal/Salvati15}
Sylvain Salvati.
\newblock {\em Lambda-calculus and formal language theory}.
\newblock Habilitation à diriger des recherches, Université de Bordeaux,
  2015.
\newblock URL: \url{https://hal.science/tel-01253426}.

\bibitem{DBLP:journals/ipl/SrivathsanW12}
B.~Srivathsan and Igor Walukiewicz.
\newblock An alternate proof of {Statman}'s finite completeness theorem.
\newblock {\em Information Processing Letters}, 112(14-15):612--616, 2012.
\newblock \href {https://doi.org/10.1016/j.ipl.2012.04.014}
  {\path{doi:10.1016/j.ipl.2012.04.014}}.

\bibitem{DBLP:journals/jsyml/Statman82}
Richard Statman.
\newblock Completeness, invariance and lambda-definability.
\newblock {\em Journal of Symbolic Logic}, 47(1):17--26, 1982.
\newblock \href {https://doi.org/10.2307/2273377} {\path{doi:10.2307/2273377}}.

\bibitem{DBLP:journals/apal/Statman04}
Richard Statman.
\newblock On the {$\lambda Y$} calculus.
\newblock {\em Annals of Pure and Applied Logic}, 130(1-3):325--337, 2004.
\newblock \href {https://doi.org/10.1016/j.apal.2004.04.004}
  {\path{doi:10.1016/j.apal.2004.04.004}}.

\bibitem{statman-dowek-1992}
Richard Statman and Gilles Dowek.
\newblock On {Statman}'s finite completeness theorem, 1992.
\newblock \href {https://arxiv.org/abs/2309.03602} {\path{arXiv:2309.03602}}.

\bibitem{Terui}
Kazushige Terui.
\newblock Semantic {Evaluation}, {Intersection} {Types} and {Complexity} of
  {Simply} {Typed} {Lambda} {Calculus}.
\newblock In {\em 23rd {International} {Conference} on {Rewriting} {Techniques}
  and {Applications} ({RTA}'12)}, pages 323--338, 2012.
\newblock \href {https://doi.org/10.4230/LIPIcs.RTA.2012.323}
  {\path{doi:10.4230/LIPIcs.RTA.2012.323}}.

\bibitem{entics:12280}
Sam van Gool, Paul-André Melliès, and Vincent Moreau.
\newblock {Profinite lambda-terms and parametricity}.
\newblock {\em {Electronic Notes in Theoretical Informatics and Computer
  Science}}, {Volume 3 -- Proceedings of MFPS XXXIX}, November 2023.
\newblock \href {https://doi.org/10.46298/entics.12280}
  {\path{doi:10.46298/entics.12280}}.

\bibitem{DBLP:journals/siglog/Walukiewicz16}
Igor Walukiewicz.
\newblock Automata theory and higher-order model-checking.
\newblock {\em {ACM} {SIGLOG} News}, 3(4):13--31, 2016.
\newblock \href {https://doi.org/10.1145/3026744.3026745}
  {\path{doi:10.1145/3026744.3026745}}.

\end{thebibliography}

\appendix

\section{The regular language of affine untyped terms}
\label{appendix:untyped-affine}

The goal of this appendix is to provide a detailed explaination of \Cref{ex:affine-terms}. We first introduce a grammar for the simply typed $\lambda$-terms of type $\UntypedTerms$ through the following notion of scoped term.

\begin{definition}
    We consider untyped terms with de Bruijn indices, given by the grammar
    \[
        u, v
        \quad::=\quad
        \var{i} \text{ for $i \in \N^*$}
        \mid
        \abs{u}
        \mid
        \app{u}{v}
    \]
    where we write $\abs{-}$ for the abstraction to distinguish it from the simply typed $\lambda$-abstraction.

    For any natural number $n$ and untyped term $u$, we define the judgment $\sc{n}{u}$ by induction with the rules
\[
    \begin{prooftree}
        \infer0[$1 \le i \le n$]{\sc{n}{\var{i}}}
    \end{prooftree}
    \qquad
    \begin{prooftree}
        \hypo{\sc{n + 1}{u}}
        \infer1{\sc{n}{\abs{u}}}
    \end{prooftree}
    \qquad
    \begin{prooftree}
        \hypo{\sc{n}{u}}
        \hypo{\sc{n}{v}}
        \infer2{\sc{n}{\app{u}{v}}}
    \end{prooftree}
    \ .
\]
The judgment $\sc{n}{u}$ has at most one derivation. We call a scoped term any pair of $n$ and $u$ such that $\sc{n}{u}$ is derivable.
\end{definition}

In the rest of the appendix, we will simply say that $\sc{n}{u}$ is a scoped term whenever this judgment is derivable. We now give the encoding of scoped terms into the simply typed $\lambda$-calculus.

\begin{definition}
    \label{def:encoding-scoped}\AP
    Let us consider a fixed sequence of simply typed variables $x_k : \tyo$ for $k \in \N$. We define the context $\Gamma_n$ as $\ell : (\tyo \tto \tyo) \tto \tyo,
        a : \tyo \tto \tyo \tto \tyo,
        x_1 : \tyo,
        \dots,
        x_n : \tyo$.

    For any natural number $n$ and untyped term $u$, we consider the encoding $\intro*\stsc{n}{u}$ of a scoped term defined by induction as
    \begin{align*}
        \stsc{n}{\var{i}}
        \quad&:=\quad
        x_{n + 1 - i}
        \\
        \stsc{n}{\abs{u}}
        \quad&:=\quad
        \ell\ (\lambda(x_{n + 1} : \tyo).\ \stsc{n + 1}{u})
        \\
        \stsc{n}{\app{u}{v}}
        \quad&:=\quad
        a\ (\stsc{n}{u})\ (\stsc{n}{v})
    \end{align*}
    which is such that $\Gamma_n \vdash \stsc{n}{u} : \tyo$.
\end{definition}

Using normalization for the simply typed $\lambda$-calculus, we can claim the following fact.

\begin{claim}
    \label{claim:bij-untyped-terms}
    For every natural number $n$, the encoding of scoped terms $\sc{n}{u}$ into the open $\lambda$-terms $\stsc{n}{u}$ of type $\tyo$ in context $\Gamma_n$, taken modulo $\beta\eta$-conversion, is bijective.

    In particular, the encoding induces a bijection between the set $\Tm{\UntypedTerms}$ and the set of closed untyped terms, i.e.\ the $u$ such that $\sc{0}{u}$ is a scoped term.
\end{claim}

\begin{definition}
    \label{def:occ-scoped}
    For any scoped term $\sc{n}{u}$ and $1 \le i \le n$, we define the natural number $\occ{i}(\sc{n}{u})$, the number of occurences of the $i^\text{th}$ variable, by induction as:
    \begin{align*}
        \occ{i}(\sc{n}{\var{j}})
        \quad&:=\quad
        \text{$1$ if $i = j$ otherwise $0$}
        \\
        \occ{i}(\sc{n}{\abs{u}})
        \quad&:=\quad
        \occ{i + 1}(\sc{n + 1}{u})
        \\
        \occ{i}(\sc{n}{\app{u}{v}})
        \quad&:=\quad
        \occ{i}(\sc{n}{u}) + \occ{i}(\sc{n}{v})
        \ .
    \end{align*}
    When the scoped term $\sc{n}{u}$ is clear from context, we will simply write $\occ{i}$.
\end{definition}

\begin{definition}
    \label{def:sem-scoped}\AP
For any finite set $Q$ and functions $\fabs : (Q \tto Q) \to Q$ and $\fapp : Q \to Q \tto Q$, we interpret any scoped term $\sc{n}{u}$ as its semantics
\[
    \sc{n}{u}
    \qquad\rightsquigarrow\qquad
    \intro*\semsc{n}{u}
    \quad:\quad
    Q^n \longto Q
\]
which is the set-theoretic function defined, for $q_1, \dots, q_n \in Q$, by induction as:
\begin{align*}
    \semsc{n}{\var{i}}[q_1, \dots, q_n]
    \quad&:=\quad
    q_i
    \\
    \semsc{n}{\abs{u}}[q_1, \dots, q_n]
    \quad&:=\quad
    \fabs(q \mapsto \semsc{n + 1}{u}[q, q_1, \dots, q_n])
    \\
    \semsc{n}{\app{u}{v}}[q_1, \dots, q_n]
    \quad&:=\quad
    \fapp(\semsc{n}{u}[q_1, \dots, q_n])(\semsc{n}{v}[q_1, \dots, q_n])
\ .
\end{align*}
We write the arguments of the function $\semsc{n}{u}$ between square brackets $[$ and $]$.
\end{definition}

\begin{remark}
    \label{rmk:sem-scoped-and-typed}
    The semantics of \Cref{def:sem-scoped} factor through the encoding of \Cref{def:encoding-scoped} in the simply typed $\lambda$-calculus and its semantic interpretation as for all $q_1, \dots, q_n \in Q$,
\[
    \semsc{n}{u}(q_1, \dots, q_n)
    \quad=\quad
    \semtm{\,\stsc{n}{u}\,}{Q}(\fabs)(\fapp)(q_n, \dots, q_1)
    \ .
\]
\end{remark}

We now instantiate \Cref{def:sem-scoped} with the following values of $Q$, $\fabs$ and $\fapp$:
\begin{itemize}
    \item $Q$ is the set $\{0, 1, \infty\} \times \{\bot, \top\}$, with its product monoid structure. For any $q \in Q$, we write $q_1 \in \{0, 1, \infty\}$ and $q_2 \in \{\bot, \top\}$ for its two components.
    \item $\fabs$ is the function $(Q \tto Q) \to Q$ defined as
    \[
        g
        \quad\longmapsto\quad
        (g(0, \top)_1
        \ ,\ 
        g(0, \top)_2 \,\wedge\, (g(1, \top)_1 \neq \infty))
    \]
    \item $\fapp$ is the curried monoid product of $Q$, i.e.\ the function $Q \to (Q \tto Q)$ defined as
    \[
        (n, b)
        \quad\longmapsto\quad
        (n', b')
        \mapsto
        (n + n', b \wedge b')
        \ .
    \]
\end{itemize}

\begin{proposition}[Left part of the tuple]
    \label{left-part}
    For any scoped term $\sc{n}{u}$ and any elements $k_1, \dots k_n$ of $\{0, 1, \infty\}$, we have
\[
    \semsc{n}{u}[(k_1, \top)\,,\,\dots\,,\,(k_n, \top)]_1
    \quad=\quad
    \occ{1} \cdot k_1 + \dots + \occ{n} \cdot k_n
\]
where the product $\cdot : \N \times \{0, 1, \infty\} \to \{0, 1, \infty\}$ comes from the monoid structure of $\{0, 1, \infty\}$.
\end{proposition}

\begin{proof}
    We verify this by induction on the scoped term $\sc{n}{u}$.
    \begin{align*}
        &\semsc{n}{\var{i}}[(k_1, \top)\,,\,\dots\,,\,(k_n, \top)]_1
        \\
        &\ :=\ 
        k_i
        \\
        &\ =\ 
        \occ{1} \cdot k_1 + \dots + \occ{n} \cdot k_n
        \\[.5em]
        &\semsc{n}{\abs{u}}[(k_1, \top)\,,\,\dots\,,\,(k_n, \top)]_1
        \\ 
        &\ :=\ 
        \fabs\left((k, b) \mapsto \semsc{n + 1}{u}[(k, b)\,,\,(k_1, \top)\,,\,\dots\,,\,(k_n, \top)]\right)_1
        \\
        &\ =\ 
        \semsc{n + 1}{u}[(0, \top)\,,\,(k_1, \top)\,,\,\dots\,,\,(k_n, \top)]_1
        \\
        &\ =\ 
        \occ{1} \cdot k_1 + \dots + \occ{n} \cdot k_n
        \\[.5em]
        &\semsc{n}{\app{u}{v}}[(k_1, \top)\,,\,\dots\,,\,(k_n, \top)]_1
        \\
        &\ :=\ 
        \fapp(\semsc{n}{u}[(k_1, \top)\,,\,\dots\,,\,(k_n, \top)])
        (\semsc{n}{v}[(k_1, \top)\,,\,\dots\,,\,(k_n, \top)])_1
        \\
        &\ =\ 
        \semsc{n}{u}[(k_1, \top)\,,\,\dots\,,\,(k_n, \top)]_1
        \ +\ 
        \semsc{n}{v}[(k_1, \top)\,,\,\dots\,,\,(k_n, \top)]_1
        \\
        &\ =\ 
        \occ{1} \cdot k_1 + \dots + \occ{n} \cdot k_n
    \end{align*}
    In the $\textabs$ case, the last equality is obtained by remarking that $\occ{1}(\sc{n + 1}{u})$ is multiplied by $0$ and that $\occ{i + 1}(\sc{n + 1}{u}) = \occ{i}(\sc{n}{\abs{u}})$ for $i \ge 1$.
\end{proof}

We introduce the following definition.

\begin{definition}
    \label{def:affine-bound-variables}\AP
    We define the property of a scoped term to be  ""affine in its bound variables"" by induction as follows:
    \begin{itemize}
        \item $\sc{n}{\var{i}}$ is always "affine in its bound variables",
        \item $\sc{n}{\abs{u}}$ is "affine in its bound variables" if and only if
        \[
            \occ{1}(\sc{n + 1}{u})
            \ \le\ 
            1
            \quad\text{and}\quad
            \text{$\sc{n + 1}{u}$ is affine in its bound variables}
        \]
        \item $\sc{n}{\app{u}{v}}$ is "affine in its bound variables" if and only if
        \[
            \text{$\sc{n}{u}$ and $\sc{n}{v}$ are both "affine in their bound variables".}
        \]
    \end{itemize}
    A $\lambda$-term $t \in \Tm{\UntypedTerms}$ will said to be ""affine"" when the closed untyped term $\sc{0}{u}$ bijectively associated to $t$ by \Cref{claim:bij-untyped-terms} is "affine in its bound variables".
\end{definition}

\begin{proposition}[Right part of the tuple]
    \label{right-part}
    For any scoped term $\sc{n}{u}$, we have
    \[
        \semsc{n}{u}[(0, \top)\,,\,\dots\,,\,(0, \top)]_2
        \quad=\quad
        b
    \]
    where $b$ is $\top$ if and only if the scoped term $\sc{n}{u}$ is "affine in its bound variables".
\end{proposition}

\begin{proof}
    We prove this by induction on the scoped term $\sc{n}{u}$.
    \begin{itemize}
    \item For any $1 \le i \le n$, $\sc{n}{\var{i}}$ is always "affine in its bound variables", and we always have
    \[
        \semsc{n}{\var{i}}[(0, \top)\,,\,\dots\,,\,(0, \top)]_2
        \ =\ 
        \top
        \ .
    \]
    \item For any scoped term $\sc{n + 1}{u}$, we have
    \begin{align*}
        &\semsc{n}{\abs{u}}[(0, \top)\,,\,\dots\,,\,(0, \top)]_2
        \\
        &\ :=\ 
        \fabs((k, b) \mapsto \semsc{n + 1}{u}[(k, b)\,,\,(0, \top)\,,\,\dots\,,\,(0, \top)])_2
        \\
        &\ =\ 
        \semsc{n + 1}{u}[(0, \top)\,,\,(0, \top)\,,\,\dots\,,\,(0, \top)]_2
        \\
        &\qquad\wedge\quad
        \semsc{n + 1}{u}[(1, \top)\,,\,(0, \top)\,,\,\dots\,,\,(0, \top)]_1 \neq \infty
        \\
        &\ =\ 
        \semsc{n + 1}{u}[(0, \top)\,,\,(0, \top)\,,\,\dots\,,\,(0, \top)]_2
        \\
        &\qquad\wedge\quad
        \occ{1}(\sc{n + 1}{u}) \le 1
        \ .
    \end{align*}
    where the last step comes from \Cref{left-part}. By the induction hypothesis on the scoped term $\sc{n + 1}{u}$, we get that $\sc{n}{\abs{u}}$ is "affine in its bound variables" if and only if $\semsc{n}{\abs{u}}[(0, \top)\,,\,\dots\,,\,(0, \top)]_2$ is $\top$.
    \item For any scoped terms $\sc{n}{u}$ and $\sc{n}{v}$, we have
    \begin{align*}
        &\semsc{n}{\app{u}{v}}[(0, \top)\,,\,\dots\,,\,(0, \top)]_2
        \\
        &\ :=\ 
        \fapp(\semsc{n}{u}[(0, \top)\,,\,\dots\,,\,(0, \top)])
        (\semsc{n}{v}[(0, \top)\,,\,\dots\,,\,(0, \top)])_2
        \\
        &\ =\ 
        \semsc{n}{u}[(0, \top)\,,\,\dots\,,\,(0, \top)]_2
        \ \wedge\ 
        \semsc{n}{v}[(0, \top)\,,\,\dots\,,\,(0, \top)]_2
    \end{align*}
    which shows, by the induction hypotheses on $\sc{n}{u}$ and $\sc{n}{v}$, that $\sc{n}{\app{u}{v}}$ is "affine in its bound variables" if and only if $\semsc{n}{\app{u}{v}}[(0, \top)\,,\,\dots\,,\,(0, \top)]_2$ is $\top$. \qedhere
    \end{itemize}
\end{proof}

\begin{theorem}
    The language of closed affine untyped terms is regular in $\FinSet$.
\end{theorem}

\begin{proof}
    We consider the language
    \[
        L
        \ :=\ 
        \{t \in \Tm{\UntypedTerms} \mid \text{$t$ is "affine"}\}
        \ .
    \]
    By definition, $t \in \Tm{\UntypedTerms}$ is "affine" if and only if $\sc{0}{u}$ is "affine in its bound variables", where $\sc{0}{u}$ is the unique scoped term such that $t = \stsc{0}{u}$, given by \Cref{claim:bij-untyped-terms}. Moreover, by \Cref{right-part}, we have that $\sc{0}{u}$ is "affine in its bound variables" if and only if
    \[
        \semsc{0}{u}_2
        \quad=\quad
        \top
    \]
    and, by \Cref{rmk:sem-scoped-and-typed}, $\semsc{0}{u}_2 = \semtm{t}{Q}(\fabs)(\fapp)$. Therefore, if we define the subset $F$ of $\semty{\UntypedTerms}{Q}$ as
    \[
        F
        \quad:=\quad
        \{s \in \semty{\UntypedTerms}{Q} \mid s(\fabs)(\fapp)_2 = \top\}
    \]
    we get that $L = \Lang{F}$ and therefore that $L$ is \salvatiReg{\FinSet}.
\end{proof}

\section{Sub-CCC proofs}
\label{appendix-sub-CCC}

Throughout the article, we construct new CCCs as full subcategories of other CCCs which verify some stability properties. A sub-CCC is a sub-category whose inclusion functor is a CCC functor. More precisely, this construction is justified by the following claim:

\begin{claim}
    \label{claim:full-sub-ccc}
    Let $\C$ be a CCC and $\D$ be a full subcategory of $\C$. If:
    \begin{itemize}
        \item the terminal object of $\C$ belongs to $\D$,
        \item the objects of $\D$ are stable under cartesian products,
        \item the objects of $\D$ are stable under exponentiation,
    \end{itemize}
    then $\D$ is a sub-CCC of $\C$.
\end{claim}

\paragraph*{Proof of \Cref{thm:squeezing-ccc}}
\label{paragraph:proof-squeezing-ccc}

By definition of a "squeezing structure", we have the following facts:
\begin{itemize}
    \item The terminal object~$1$ of~$\C$ is an object of~$\Sqz{\C}$.
    \item If~$c$ and~$c'$ belong to~$\Sqz{\C}$, then~$c \times c'$ still does. Indeed, we have the following morphisms:
    \[
        L_{c \times c'} \longtoleft L_c \times L_{c'} \longtoleft c \times c'
        \qquand
        c \times c' \longtoright R_c \times R_{c'} \longtoright R_{c \times c'}
        \ .
    \]
    \item If~$c$ and~$c'$ belong to~$\Sqz{\C}$, then~$c \tto c'$ still does. Indeed, we have the following morphisms:
    \[
        L_{c \tto c'} \longtoleft R_c \tto L_{c'} \longtoleft c \tto c'
        \qquand
        c \tto c' \longtoright L_c \tto R_{c'} \longtoright R_{c \tto c'}
        \ .
    \]
\end{itemize}
By \Cref{claim:full-sub-ccc}, we get that $\Sqz{\C}$ is a sub-CCC of $\C$.

\paragraph*{Proof of \Cref{prop:ge-ccc}}
\label{paragraph:proof-ge-ccc}

The following facts hold about the category $\Cge$ of "inhabited objects":
\begin{itemize}
    \item The terminal object $1$ of $\C$ is in $\Cge$.
    \item If $c$ and $c'$ belong to $\Cge$, which means that there exists $f \in \C(1, c)$ and $f' \in \C(1, c')$, then the existence of the morphism $\langle f, f'\rangle \in \C(1, c \times c')$ shows that $c \times c'$ belongs to $\Cge$.
    \item If $c$ is any object of $\C$ (so, it may in particular be an object of $\Cge$) and $c'$ is an object of $\Cge$, which means that there exists a morphism $f' \in \C(1, c')$, then the curryfication of the composition $f \circ !_c$, where $!_c \in \C(c, 1)$, shows that $c \To c'$ belongs to $\Cge$. Notice that the argument actually shows that $\Cge$ is an exponential ideal.
\end{itemize}
By \Cref{claim:full-sub-ccc}, this proves that $\Cge$ is a sub-CCC of $\C$.

\paragraph*{First proof of \Cref{prop:ccc-partial-surj}}
\label{paragraph:first-proof-ccc-partial-surj}

Let $R$ be the "CCC of logical relations" from $\FinSet$ to $\E$. We show that the full subcategory $\Rps$ of objects $(Q, e, \Vdash)$ such that the relation $\Vdash$ is a partial surjection is a sub-CCC of $\R$. We have the following facts:
\begin{itemize}
    \item The terminal object of $\R$ belongs to $\Rps$.
    \item Let $(Q, e, \Vdash)$ and $(Q', e', \Vdash')$ be objects of $\Rps$, and $\Vdash^\times$ be the relation of their product.
    \begin{itemize}
        \item Let $\langle q, q'\rangle \in Q \times Q'$ and $\langle x_i, x_i'\rangle \in \E(1, e \times e')$ such that $\langle q, q'\rangle \Vdash^\times \langle x_i, x_i'\rangle$ for $i = 1, 2$.
        
        This means that $q \Vdash x_i$ and $q' \Vdash' x_i'$ for $i = 1, 2$.
        By functionality of $\Vdash$ and $\Vdash'$, we get that $x_1 = x_2$ and $x_1' = x_2'$.
        Therefore, $\langle x_1, x_1'\rangle = \langle x_2, x_2'\rangle$, so $\Vdash^\times$ is functional.
        
        \item Let $\langle x, x'\rangle$ be any element of $\E(1, e \times e')$.

        By surjectivity of $\Vdash$ and $\Vdash'$, there exists $q \in Q$ and $q' \in Q'$ such that $q \Vdash x$ and $q \Vdash x'$.
        Therefore, $\langle q, q'\rangle \Vdash^\times \langle x, x'\rangle$, so $\Vdash^\times$ is surjective.
    \end{itemize}
    \item Let $(Q, e, \Vdash)$ and $(Q', e', \Vdash')$ be objects of $\Rps$, and $\Vdash^\tto$ be the relation of their arrow.
    \begin{itemize}
        \item Let $f \in Q \tto Q'$ and $g_i \in \E(1, e \tto e')$ such that $f \Vdash^\tto g_1$ and $f \Vdash^\tto g_2$.
        
        Let $x$ be any element of $\E(1, e)$.
        
        By surjectivity of $\Vdash$, there exists a $q \in Q$ such that $q \Vdash x$, so $f(q) \Vdash' g_i \circ x$ for $i = 1, 2$.
        
        By functionality of $\Vdash'$, we get $g_1 \circ x = g_2 \circ x$. As this equality holds for all $x \in \E(1, e)$ \emph{and $\E$ is well-pointed}, we get that $g_1 = g_2$, which proves that $\Vdash^\tto$ is functional.
        \item Let $g \in \E(1, e \tto e')$. By surjectivity of $R'$, we can choose a function $\rho\colon \E(1,X') \to Q'$ such that $(\rho(x'), x') \in R'$ for every $x'$. Let $f'(q) = \rho(g \circ (R(q)))$ where $R$, as a functional relation, can seen as a partial function $Q \rightharpoonup \E(1,X)$. Then $f'$ is also a partial function; since its codomain $Q'$ is non-empty (indeed, $\E(1,X')\neq\varnothing$ by assumption and $\rho(\E(1,X')) \subseteq Q'$) it can be extended to a total function $f\colon Q\to Q'$. This construction ensures that $(f,g)\in R \tto R'$, so we have shown that $R \tto R'$ is surjective. \qedhere
    \end{itemize}
\end{itemize}
By \Cref{claim:full-sub-ccc}, we get that $\Rps$ is a sub-CCC of $\R$.

\paragraph*{Proof of \Cref{prop:le-eq-ccc}}
\label{paragraph:proof-le-eq-ccc}

We remark that the following facts hold about $\Ele$:
\begin{itemize}
    \item The terminal object $1$ of $\E$ is in $\Ele$.
    \item If $e$ and $e'$ belong to $\Ele$, then $\E(1, e \times e')$ is in bijection with $\E(1, e) \times \E(1, e')$, so $e \times e'$ belongs to $\Ele$.
    \item If $e$ and $e'$ belong to $\Ele$, as $\E$ is well-pointed, we have an injection of $\E(1, e \tto e')$ into $\E(1, e) \tto \E(1, e')$, so $e \tto e'$ belongs to $\Ele$.
\end{itemize}
By \Cref{claim:full-sub-ccc}, this proves that $\Ele$ is a sub-CCC of $\E$.

Let $e$ and $e'$ be two objects of $\Eeq$. As the object $e \tto e'$ belongs to both~$\Ege$ and~$\Ele$, there exists a unique morphism from $e$ to $e'$. This shows that $\Eeq$ is equivalent to the terminal category.

\section{Partial surjections}
\label{appendix:partial-surjections}

\paragraph*{Proof of \Cref{lem:sqz-partial-surj}}
\label{appendix:proof-sqz-partial-surj}

Let $c = (c_1, c_2, \Vdash)$ be an object of $\Sqz{\R}$. There exist morphisms
\[
    (u_1, u_2)
    \ :\ 
    L_c \longtoleft c
    \qquand
    (v_1, v_2)
    \ :\ 
    c \longtoright R_c
    \ .
\]
We show that the relation $\Vdash$ is functional and surjective.
\begin{itemize}
    \item Let $x_1 \in \C_1(1, c_1)$ and $x_2, x_2' \in \C_2(1, c_2)$ such that $x_1 \Vdash x_2$ and $x_1 \Vdash x_2'$.
    
    As $(v_1, v_2)$ is a morphism, we get $v_1 \circ x_1 \Vdash v_2 \circ x_2$ and $v_1 \circ x_1 \Vdash v_2 \circ x_2'$.
    As the relation of $R_c$ is functional, we get that $v_2 \circ x_2 = v_2 \circ x_2'$.
    By injectivity of $\C_2(1, v_2)$, we get that $x_2 = x_2'$, so the relation $\Vdash$ is functional.
    \item Let $x_2 \in \C_2(1, c_2)$. We write $(c_{L, 1}, c_{L, 2} \Vdash_L)$ for the object $L_c$.
    
    By surjectivity of $\C_2(1, v_2)$, there exists $x_{L, 2} \in \C_2(1, c_{L, 2})$ such that $x_2 = v_2 \circ x_{L, 2}$.
    By surjectivity of the relation $\Vdash_L$, there exists $x_{L, 1}$ such that $x_{L, 1} \Vdash_L x_{L, 2}$.
    As $(u_1, u_2)$ is a morphism, we have $u_1 \circ x_{L, 1} \Vdash x_2$, so the relation $\Vdash$ is surjective.
\end{itemize}
This shows that all objects in $\Sqz{\R}$ are partial surjections.

\section{Squeezing structures}

\paragraph*{Second proof of \Cref{prop:ccc-partial-surj}}
\label{paragraph:second-proof-ccc-partial-surj}

There exists a "squeezing structure" such that
\begin{itemize}
    \item for any object $R = (Q, e, \Vdash)$ of $\R$, the objects $L_R$ and $R_R$ are both equal to $T_e$;
    \item the two wide subcategories $\Rleft$ and $\Rright$ are both taken to be the wide subcategory of "target-identities".
\end{itemize}
The fact that the point functor $T$ is product-preserving gives us all the "target-identities" of the "squeezing structure", except for the case of the morphism
\[
    (\E(1, e), e, \sim_e) \tto (\E(1, e'), e', \sim_{e'})
    \ \longto\ 
    (\E(1, e \tto e'), e \tto e', \sim_{e \tto e'})
    \ .
\]
For this morphism, we will crucially use the fact that we relate $\FinSet$ and $\E$, which is well-pointed. We have a set-theoretic function
\[
    i
    \quad:\quad
    \E(1, e \tto e')
    \ \longto\ 
    \E(1, e) \tto \E(1, e')
\]
which is injective as $\E$ is well-pointed and lifts to a "target-identity". As the object $e \tto e'$ of $\E$ is "inhabited", the set $\E(1, e \tto e')$ is non-empty and the set-theoretic function $i$ admits a retraction
\[
    r
    \quad:\quad
    \E(1, e) \tto \E(1, e')
    \ \longto\ 
    \E(1, e \tto e')
\]
which lifts to a "target-identity" $T_e \tto T_{e'} \to T_{e \tto e'}$.
By \Cref{lem:sqz-partial-surj}, we know that the objects of $\Sqz{\R}$ are partial surjections. Conversely, suppose that $R = (Q, e, \Vdash)$ is such that $\Vdash$ is a partial surjection. We the two following "target-identities":
\begin{itemize}
    \item The fact that the relation $\Vdash$ is surjective yields a set-theoretic function $\E(1, e) \to Q$ which lifts to a "target-identity" $T_e \to R$.
    \item As the relation $\Vdash$ is functional and $\E(1, e)$ is non-empty as $e$ is "inhabited", there exists a set-theoretic function $Q \to \E(1, e)$ extending $\Vdash$, and any such set-theoretic function lifts to a "target-identity" $R \to T_e$.
\end{itemize}
This shows that the objects of $\Sqz{\R}$ are exactly the partial surjections, which then form a sub-CCC of $\R$ by \Cref{thm:squeezing-ccc}.

\paragraph*{Proof of \Cref{prop:squeezing-lam-finord}}
\label{paragraph:proof-squeezing-lam-finord}

It is clear that "target-identities" are composable and stable under finite products and exponentials, which is what is asked given that left and right morphisms are the same in the present case.

We are left with the task to show the existence of the morphisms as described in \Cref{eq:left-right-morphisms} in our particular setting. As there exists at most one "target-identity" whose $\Lam$ component is a given $\lambda$-term, we give these $\lambda$-terms.
\begin{itemize}
    \item \textbf{Case~$\Bij{1} \to 1$:} The unique morphism $\Bij{1} \to 1$ is a "target-identity".
    \item \textbf{Case~$1 \to \Bij{1}$:} The $\lambda$-term $\lambda (y : 1).\,\lambda(x : \tyo).\,x$ lifts to a "target-identity"
    \[
        1
        \quad\longto\quad
        (\Fin{1}, 1, \logrel{1})
        \ .
    \]
    \item \textbf{Case~$\Bij{n \times n'} \to \Bij{n} \times \Bij{n'}$:} the fact that~$\Bij{(-)}$ is a functor gives us directly such a morphism which can be verified to be a "target-identity".
    \item \textbf{Case~$\Bij{n} \times \Bij{n'} \to \Bij{n \times n'}$:} The morphism is given by the~$\lambda$-term
    \[
        \lambda(p : \Fin{n} \times \Fin{n'}).
        \,\lambda(x : \tyo^{n \times n'}).\,
        p_1\ 
        \langle\,
            \Fin{1 \times \Id}\ p_2\ x
            , \dots, 
            \Fin{n \times \Id}\ p_2\ x
        \rangle
    \]
    where~$i \times \Id$ is the function~$n' \to n \times n'$ sending $j$ on $(i, j)$, from which we get the~$\lambda$-term~$\Fin{i \times \Id}$ of simple type~$\Fin{n'} \tto \Fin{n \times n'}$.
    
    This $\lambda$-term lifts to a "target-identity"
    \[
        \Bij{n} \times \Bij{n'}
        \ \longto\ 
        \Bij{n \times n'}
        \ .
    \]
    \item \textbf{Case~$\Bij{n \tto n'} \to \Bij{n} \tto \Bij{n'}$:} We have the "target-identity"
    \[
        \Bij{n \tto n'} \times \Bij{n}
        \ \longto\ 
        \Bij{(n \tto n') \times n}
    \]
    which, when composed with the morphism $\Bij{\ev_{n, n'}}$
    which has the evaluation morphism~$\ev_{n, n'}$ as target-component, yields a morphism
    \[
        \Bij{n \tto n'} \times \Bij{n}
        \ \longto\ 
        \Bij{n'}
    \]
    which, after curryfication, gets us a "target-identity" $\Bij{n \tto n'} \to \Bij{n} \tto \Bij{n'}$.
    \item \textbf{Case~$\Bij{n} \tto \Bij{n'} \to \Bij{n \tto n'}$:} Notice that the equality
    \[
        n \tto n'
        \ =\ 
        \underbrace{n' \times \dots \times n'}_{\text{$n$ times}}
    \]
    shows that the~$\lambda$-term of type~$\Fin{n} \tto \Fin{n'} \to (\Fin{n'})^n$
    \[
        \lambda(F : \Fin{n} \tto \Fin{n'}).\,\langle F\,\pi_{1}, \dots, F\,\pi_n \rangle
    \]
    lifts to a "target-identity" $\Bij{n} \tto \Bij{n'} \to \Bij{n'}^n$. By postcomposing this "target-identity" with an iteration of the "target-identities"~$\Bij{m} \times \Bij{m'} \to \Bij{m \times m'}$, we obtain the "target-identity"
    \[
        \Bij{n} \tto \Bij{n'}
        \ \longto\ 
        \Bij{n \tto n'}
        \ .
    \]
\end{itemize}
This finishes the proof that there is a "squeezing structure" as described in the statement of the proposition.

\end{document}